\documentclass[aps,pra,twocolumn,showpacs,showkeys,superscriptaddress,nobalancelastpage,nofootinbib,floatfix,longbibliography]{revtex4-2}

\usepackage{amsmath,amssymb,amsthm}
\usepackage{graphicx}
\usepackage{physics}
\usepackage{enumitem}
\usepackage{xspace}

\usepackage{orcidlink}

\usepackage{url}
\usepackage{hyperref}
\usepackage{color}
\definecolor{refcolor}{RGB}{160,35,0}
\definecolor{hrefcolor}{RGB}{0,35,190}
\hypersetup{
    colorlinks,
    citecolor=refcolor,
    filecolor=refcolor,
    linkcolor=hrefcolor,
    urlcolor=hrefcolor
}

\def\({\left(}
\def\){\right)}
\def\[{\left[}
\def\]{\right]}

\newcommand{\hilbert}{\mathcal{H}}
\newcommand{\pobs}[1]{#1}
\newcommand{\obs}[1]{\mathsf{\pobs{#1}}}

\newcommand{\mc}[1]{\mathcal{#1}}

\newcommand{\tn}[1]{\textnormal{#1}}
\newcommand{\R}{\mathbb{R}}
\newcommand{\C}{\mathbb{C}}
\newcommand{\N}{\mathbb{N}}
\newcommand{\Z}{\mathbb{Z}}
\newcommand{\WdWSpaceL}{{\mkern-8mu}}
\newcommand{\WdWSpaceM}{{\mkern-6mu}}
\newcommand{\WdWSpaceMM}{{\mkern-9mu}}
\newcommand{\WdWSpaceR}{{\mkern-8mu}}
\newcommand{\braa}[1]{\bra{\WdWSpaceL\bra{#1}\WdWSpaceM}}
\newcommand{\kett}[1]{\ket{\WdWSpaceM\ket{#1}\WdWSpaceR}}
\newcommand{\braakett}[2]{\expval{\WdWSpaceL\bra{#1}{\WdWSpaceMM}\ket{#2}{\WdWSpaceR}}}
\newcommand{\brakett}[2]{\left.\bra{#1}{\WdWSpaceMM}\ket{#2}{\WdWSpaceR}\right\rangle}

\newcommand{\ie}{\textit{i.e.}\xspace}
\newcommand{\vs}{\textit{vs.}\xspace}
\newcommand{\eg}{\textit{eg.}\xspace}

\newcommand{\etal}{\textit{et al.}\xspace}

\newcommand{\sref}[1]{\S\ref{#1}}

\newtheorem{theorem}{Theorem}
\newtheorem{definition}{Definition}
\newtheorem{proposition}{Proposition}
\newtheorem{corollary}{Corollary}
\newtheorem{hypothesis}{Hypothesis}
\newtheorem{example}{Example}

\theoremstyle{remark}
\newtheorem{remark}{Remark}

\newtheoremstyle{nospace} 
  {3pt}                   
  {3pt}                   
  {\itshape}              
  {}                      
  {\bfseries}             
  {.}                     
  {.5em}                  
  {\thmname{#1}\thmnumber{#2}\thmnote{ (#3)}} 

\theoremstyle{nospace}

\newtheorem{defCustom}{}
\makeatletter
\newcommand{\setdefCustomtag}[1]{
  \let\oldthedefCustom\thedefCustom
  \renewcommand{\thedefCustom}{{\normalfont\textbf{#1}}}
  \g@addto@macro\enddefCustom{
    \global\let\thedefCustom\oldthedefCustom}
  }
\makeatother

\allowdisplaybreaks

\begin{document}

\title{The clock ambiguity problem: extended or extinguished?}
\author{\orcidlink{0000-0002-2765-1562}\ O.C. Stoica}
\affiliation{Dept. Th. Physics, NIPNE-HH, Bucharest, Romania. \href{mailto:cristi.stoica@theory.nipne.ro}{cristi.stoica@theory.nipne.ro},  \href{mailto:holotronix@gmail.com}{holotronix@gmail.com}}

\begin{abstract}
I show that the clock ambiguity cannot be solved by a purely relational condition like the noninteraction condition, and it is even stronger, extending to evolution laws. The ambiguity is solved by specifying the physical meaning of observables.

Page and Wootters (1983) showed how time and dynamics can emerge from entanglement within a stationary quantum system containing a clock. The clock ambiguity problem is that, from a purely relational stance and without fixing a clock-world split, the emergence is ambiguous, resulting in any possible history (Albrecht 1995).

I show that the ambiguity is stronger than previously recognized. Under the relational stance, it extends from histories to the evolution laws themselves. The spectrum of any ideal clock uniformizes the spectra of the world's evolution operators, leaving only the dimension of the Hilbert spaces as invariant information.

Fixing the clock-world split can solve the ambiguity, but this would block spacetime symmetries.

One might want to remove the ambiguity up to a unitary equivalence by imposing noninteraction, as in Marletto and Vedral (2017). But once the clock spectrum uniformizes the world spectrum, unitary equivalence becomes too coarse to distinguish any two possible world dynamics, which is the result proved here. Thus a purely relational condition such as noninteraction is insufficient.

Nor can all different decompositions be regarded as equally valid perspectives, since then records would not be correlated with the events they record, and empirical knowledge would be impossible. The resolution is therefore not to embrace the ambiguity, but to recognize what a bare reading of the Page-Wootters structure omits: the physical meaning of the operators.

\end{abstract}

\maketitle

\begin{quote}
\textit{A man with a watch knows what time it is.\\A man with two watches is never sure.}
\end{quote}
\begin{flushright}
Segal's Law~\cite{Bloch2003MurphysLaw}
\end{flushright}

\section{Introduction}
\label{s:intro}

The \emph{problem of time} appeared in canonical quantum gravity, where the Wheeler-DeWitt constraint equation implies that the universe is stationary~\cite{Dewitt1967QuantumTheoryOfGravityI_TheCanonicalTheory,Dewitt1967QuantumTheoryOfGravityII_TheManifestlyCovariantTheory},
\begin{equation}
\label{eq:PW-WDW}
\obs{H}\kett{\Psi}=0.
\end{equation}

This problem is more general, since just like the long range of the electromagnetic interactions requires the existence of a charge superselection rule~\cite{WickWightmanWigner1952TheIntrinsicParityOfElementaryParticles}, gravitational coupling requires a similar superselection rule for energy. From this observation, Page and Wootters~\cite{PageWootters1983EvolutionWithoutEvolution} concluded that
\emph{``only operators that commute with the Hamiltonian can be observables''}, and immediately asked \emph{``But such operators are stationary, so how is it that we observe time dependence in the world?''}.

The stationarity of observables implies that observations cannot distinguish the state of the universe from a stationary state, so we need to understand how it is possible to observe changes in a world satisfying, for a general state $\rho$, the stationary equation $[\rho,\obs{H}]=0$.

The \emph{Page-Wootters proposal} is that time emerges from the stationary state $\kett{\Psi}$, if we assume that it contains a clock subsystem entangled with the rest of the world so that $\kett{\Psi}=\int_\R\ket{\tau}_c\ket{\psi(\tau)}_r\dd \tau$~\cite{PageWootters1983EvolutionWithoutEvolution,Wootters1984TimeReplacedByQuantumCorrelations,Page1986DensityMatrixOfTheUniverse,Page1989TimeAsAnInaccessibleObservable,Page1994ClockTimeAndEntropy}.
Then, for each clock state $\ket{\tau}_c$, there is a corresponding world state $\ket{\psi(\tau)}_r$, even though the total entangled state is stationary and contains all moments of time. To find the value of an observable at a given time, one has to condition on the value indicated by the clock, interpreted as time.

Due to the explanatory power of the superselection rules and their inevitability as a logical consequence of quantum theory, the problem of time is much more general than it may seem.
This should make the Page-Wootters proposal to be of much wider interest.

However, the Page-Wootters proposal was challenged by Kucha\v{r}'s criticism~\cite{Kuchar1992TimeAndInterpretationsOfQuantumGravity}, in particular for the \emph{multiple choice problem}, and for the more general \emph{clock ambiguity problem} noticed by Albrecht~\cite{Albrecht1995TheoryOfEeverythingVsTheoryOfAnything} and developed together with Iglesias~\cite{AlbrechtIglesias2008ClockAmbiguityAndTheEmergenceOfPhysicalLaws,AlbrechtIglesias2012ClockAmbiguityImplicationsNewDevelopments}.

Albrecht and Iglesias~\cite{Albrecht1995TheoryOfEeverythingVsTheoryOfAnything,AlbrechtIglesias2008ClockAmbiguityAndTheEmergenceOfPhysicalLaws,AlbrechtIglesias2012ClockAmbiguityImplicationsNewDevelopments} noticed that if we simply start with a timeless system obeying equation~\eqref{eq:PW-WDW} and try to recover the state and its time evolution $\ket{\psi(\tau)}_r$, we can get all possible states following all possible Schr\"odinger dynamics or even having a random time evolution, depending on the choice of the clock's degrees of freedom.

The very title of Albrecht's article~\cite{Albrecht1995TheoryOfEeverythingVsTheoryOfAnything}, \emph{``The theory of everything vs the theory of anything''}, seems to signal a worry about the predictive power of such a theory.
For this reason, despite the hope of Albrecht and Iglesias to gain \emph{``crucial properties of the physical world such as space, locality, gravity, gauge symmetries, and cosmology as emergent and approximate''}, maybe even a theory of everything, from as little as the equation~$\obs{H}\kett{\Psi}=0$ and physical reasoning, this ambiguity was seen as a problem of the Page-Wootters proposal.

Marletto and Vedral gave an important response to this problem, showing that, under certain conditions, the ambiguity disappears up to an equivalence once one requires an ideal clock to be noninteracting with the rest of the world~\cite{MarlettoVedral2017EvolutionWithoutEvolutionAndWithoutAmbiguities}.

The central point of this article is not to show a problem in the Page-Wootters framework, but in its reading as a bare structure from which an unambiguous description of physical reality can emerge. I show that, by taking the physical meanings of the observables into account, the clock ambiguity problem vanishes, but by adopting a purely relational stance, not only we obtain the original clock ambiguity problem, but even a stronger version, which extends from the impossibility to distinguish histories from random histories in the bare Page-Wootters framework, to the impossibility to distinguish even different evolution laws.

In the argument below, I temporarily adopt the relational stance in order to push it to its logical, but ultimately unacceptable, consequences. I then show that, if we do what quantum mechanics ordinarily does, namely keep track of the physical meanings of the observables, the Page-Wootters proposal recovers a determinate physical description.

In Section~\ref{s:reviews-two-clock-ambiguities} I will distinguish between the original ambiguity, which is between histories following any evolution, and a stronger ambiguity, which can't distinguish between the evolution laws in addition to histories.

In Section~\ref{s:ambiguity} I extend and strengthen Albrecht and Iglesias' result to the stronger kind of ambiguity. 
This stronger form of ambiguity is closest to the issue addressed by Marletto and Vedral. I show that, for ideal clocks, noninteraction by itself does not select a unique rest Hamiltonian: the ambiguity extends to Hamiltonians and becomes maximal within the bare PW structure.
The technical reason is that the spectrum of the clock's Hamiltonian washes out any information contained in the spectrum of the world's Hamiltonian.
The ambiguity shows that all we can get from a system satisfying $\obs{H}\kett{\Psi}=0$ containing an ideal clock that is not singled out among other possible choices is the dimension of the world's Hilbert space.
This applies to both discrete and continuous ideal clocks.

In Section~\ref{s:reviews-Marletto-Vedral} I give my explanation of the apparent tension between the solution of the clock ambiguity problem offered by Marletto and Vedral and the results proved here. I track the explanation to the spectral-washing effect that ideal clocks have, which unites the equivalence classes of world Hamiltonians into a single equivalence class, leading to the maximal ambiguity proved here.

In Section~\ref{s:physical-meaning} I show that, once we keep track of which operator represents each observable physical property, the ambiguity is reduced to the freedom allowed and required by physical symmetries. For example, the position and momentum are represented by specific operators, and not by any other operators with the same spectral properties and satisfying the same commutation relations, and even providing the same form of the Hamiltonian.
Taking this into account reduces the ambiguity under the limits allowed and required by spacetime and other physical symmetries.

In Section~\ref{s:physical-ambiguity} I argue that some ambiguity is not only harmless but necessary: completely removing it, for example by fixing the clock-world decomposition too rigidly, would introduce an absolute time and block relativity of simultaneity, diffeomorphism invariance, and even Galilean relativity.
But maximal ambiguity is not good either, as it removes any chance of predictive or explanatory power, since all possible laws and histories would manifest in the same way.
In particular, this leaves no room for emergence programs like the one Albrecht and Iglesias hoped to construct using the freedom provided by the clock ambiguity.

One may still hope that we can avoid assigning predefined meanings as physical properties to the operators, as this may seem to be redundant baggage, and that these physical meanings emerge instead from the spectral properties and the mutual relations between the operators (Section~\ref{s:ambiguity-origin}).
In Section~\ref{s:no-accept-ambiguity} I examine this hypothesis and show that if this were true, it would make it impossible even to record the results of experiments or to maintain memories about events, making both life and empirical science impossible.

\section{The Page-Wootters formalism for continuous and discrete time}
\label{s:reviews-Page-Wootters}

Just like in the case of positions~\cite{AharonovKaufherr1984QuantumFramesOfReference,BartlettRudolphSpekkens2007ReferenceFramesSuperselectionRulesAndQuantumInformation}, we never observe changes with respect to an absolute time, the only change we observe is by comparing relative changes between systems.
But somehow, from this web of relations, an apparently universal time seems to emerge. Page and Wootters showed that this can be modeled by the existence of a subsystem playing the role of a clock, so that the total state space $\hilbert$ can be divided in a clock subsystem $\hilbert_c$ and a ``rest of the world'' subsystem $\hilbert_r$,
\begin{equation}
\label{eq:hilbert-cr}
\hilbert=\hilbert_c\otimes\hilbert_r.
\end{equation}

Then, the succession of states of the world in time can be decoded from a timeless entanglement between the world and the clock~\cite{PageWootters1983EvolutionWithoutEvolution,Page1986DensityMatrixOfTheUniverse,Page1994ClockTimeAndEntropy,Wootters1984TimeReplacedByQuantumCorrelations}, as I will describe now.

\begin{definition}
\label{def:clock}
An \emph{ideal clock} is a system with states $\ket{\tau}_c$ and transitions $\obs{U}_c(\tau)$, $\tau\in\R$, so that
\begin{equation}
\label{eq:clock-evol}
\ket{\tau'}_c=\obs{U}_c(\tau-\tau')\ket{\tau}_c.
\end{equation}

An ideal clock can be represented formally by a quantum system with Hilbert space $\hilbert_c\cong L^2(\R)$ and generalized clock-reading vectors $\ket{\tau}_c$ satisfying
\begin{equation}
\label{eq:t-basis}
\braket{\tau}{\tau'}_c=\delta(\tau-\tau')
\end{equation}
and a Hamiltonian operator
\begin{equation}
\label{eq:clock-H}
\obs{H}_c=-i\hbar\pdv{\tau}
\end{equation}
so that $\obs{U}_c(\tau)=e^{-\frac{i}{\hbar}\obs{H}_c\tau}$.
The \emph{time operator}, defined as $\obs{T}_c:=\int_{\R}\tau\dyad{\tau}_c\dd \tau$, is canonically conjugate to $\obs{H}_c$, \ie
\begin{equation}
\label{eq:clock-time-operator}
[\obs{T}_c,\obs{H}_c]=i\hbar.
\end{equation}

The ideal clock readings are represented by the generalized eigenvectors
$\ket{\tau}_c$ of $\obs{T}_c$. In the idealized PW construction, the
relative states are obtained by conditioning on these clock readings.
\end{definition}

Often less ideal kinds of clocks are considered, for example clocks that interact weakly with the rest of the world were already mentioned in~\cite{PageWootters1983EvolutionWithoutEvolution}, and also clocks with a finite dimensional Hilbert space $\hilbert_c$ and discrete time~\cite{Albrecht1995TheoryOfEeverythingVsTheoryOfAnything}, and even with unsharp time values~\cite{Busch2002Time-energy-uncertainty-relation,LoveridgeMiyadera2019RelativeQuantumTime,SmithAhmadi2019QuantizingTimeInteractingClocksAndSystems,HohnSmithLock2021TrinityOfRelationalQuantumDynamics}.

In the following we will write informally $\ket{\tau}_c\in\hilbert_c$ and call it an eigenvector of $\obs{T}_c$ (with $\obs{T}_c\ket{\tau}_c=\tau\ket{\tau}_c$) even though $\ket{\tau}_c$ is not normalizable and so it does not belong to $\hilbert_c$, but to its rigged Hilbert space, so it is a generalized eigenvector~\cite{GelfandVilenkin1964GeneralizedFunctionsAApplicationsOfHarmonicAnalysis} satisfying condition~\eqref{eq:t-basis}.

\begin{definition}
\label{def:PW-system}
A \emph{Page-Wootters (PW) system} consists of an ideal clock and a ``rest of the world'' quantum system $\hilbert_r$, so that the two systems do not interact,
\begin{equation}
\label{eq:PW-noninteracting}
\obs{H}=\obs{H}_c\otimes\obs{I}_r+\obs{I}_c\otimes\obs{H}_r.
\end{equation}

Its \emph{relative temporal states} have the form
\begin{equation}
\label{eq:PW-state-temporal}
\ket{\Psi(\tau)}=\ket{\tau}_c\ket{\psi(\tau)}_r,
\end{equation}
where $\ket{\tau}_c$ is a generalized eigenvector of $\obs{T}_c$ and $\ket{\psi(\tau)}_r\in\hilbert_r$.
We interpret a state $\ket{\tau}_c\ket{\psi(\tau)}_r$ as evolving into $\ket{\tau'}_c\ket{\psi(\tau')}_r$ if $\ket{\tau'}_c=\obs{U}_c(\tau'-\tau)\ket{\tau}_c$ and
\begin{equation}
\label{eq:PW-evol}
\ket{\psi(\tau')}_r=e^{-\frac{i}{\hbar}\obs{H}_r(\tau'-\tau)}\ket{\psi(\tau)}_r,
\end{equation}
\ie the \emph{history} $\tau\mapsto\ket{\psi(\tau)}_r$ is compatible with the Hamiltonian $\obs{H}_r$.
The allowed \emph{timeless states} of the PW system have the form
\begin{equation}
\label{eq:PW-state-timeless}
\kett{\Psi}=\int_\R\ket{\tau}_c\ket{\psi(\tau)}_r\dd \tau,
\end{equation}
where the integral is taken over the history of a temporal state.
Due to their invariance to time translations, timeless states satisfy the Wheeler-DeWitt type equation~\eqref{eq:PW-WDW}.
\end{definition}

Page and Wootters proposed that the universe is stationary and in a timeless state $\kett{\Psi}$ from equation~\eqref{eq:PW-state-timeless}, which is an entangled state.
The state of the rest of the world relative to the clock's state $\ket{\tau}_c$ is
\begin{equation}
\label{eq:relative-state}
\ket{\tau}_c\mapsto\ket{\psi(\tau)}_r,
\end{equation}
and it can be understood as the world's state at time $\tau$.

If $\kett{\Psi}$ from equation~\eqref{eq:PW-state-timeless} satisfies~equation~\eqref{eq:PW-WDW}, by applying it to the Hamiltonian from equation~\eqref{eq:PW-noninteracting} we can obtain the Schr\"odinger equation for $\ket{\psi(\tau)}_r$~\cite{Wootters1984TimeReplacedByQuantumCorrelations},
\begin{equation}
\label{eq:PW-schrod-deriv}
\begin{aligned}
i\hbar\dv{\tau}\ket{\psi(\tau)}_r
&=i\hbar\dv{\tau}\bra{\tau'}{e^{\frac{i}{\hbar}\obs{H}_c(\tau-\tau')}\otimes\obs{I}_r}\kett{\Psi} \\
&=-\bra{\tau}{\obs{H}_c\otimes\obs{I}_r}\kett{\Psi} \\
&=\bra{\tau}{\obs{I}_c\otimes\obs{H}_r-\obs{H}}\kett{\Psi} \\
&=\obs{H}_r\ket{\psi(\tau)}.
\end{aligned}
\end{equation}

The \emph{conditional expectation value} of a stationary observable $\obs{A}=\obs{I}_c\otimes\obs{A}_r$ given that the clock reads $\tau$ is
\begin{equation}
\label{eq:exp-val}
E(\obs{A}|\tau)=\frac{\tr(\obs{A}\obs{P}_\tau\rho)}{\tr\obs{P}_\tau\rho},
\end{equation}
where $\obs{P}_\tau:=\dyad{\tau}\otimes\obs{I}_r$.
Then, the conditional expectation value of $\obs{A}$ evolves (in the Heisenberg and in the Schr\"odinger pictures) into~\cite{PageWootters1983EvolutionWithoutEvolution}
\begin{equation}
\label{eq:exp-val-t}
\begin{array}{ll}
E(\obs{A}|\tau)
&=\tr_r\big(\underbrace{e^{\frac{i}{\hbar}\obs{H}_r\tau}\obs{A}e^{-\frac{i}{\hbar}\obs{H}_r\tau}}_{\text{Heisenberg picture}}\rho_r(0)\big) \\
&=\tr_r\big(\obs{A}\underbrace{e^{-\frac{i}{\hbar}\obs{H}_r\tau}\rho_r(0) e^{\frac{i}{\hbar}\obs{H}_r\tau}}_{\text{Schr\"odinger picture}}\big) \\
\end{array}
\end{equation}
where $\rho_r(0):=\displaystyle{\frac{\tr_c\obs{P}_0\rho}{\tr\obs{P}_0\rho}}$.
The Schr\"odinger evolution of the density operator of the rest of the world $\rho_r(\tau)$ is then
\begin{equation}
\label{eq:PW-schrod-rho}
\rho_r(\tau)=e^{-\frac{i}{\hbar}\obs{H}_r\tau}\rho_r e^{\frac{i}{\hbar}\obs{H}_r\tau},
\end{equation}
consistent with equation~\eqref{eq:PW-schrod-deriv} for $\rho_r(\tau)=\dyad{\psi(\tau)}_r$.

The Page-Wootters proposal also works with finite or infinite discrete clocks. Then, $\tau\in\Z_n$, where $n$ is either a positive integer, or infinite, when $\Z_\infty:=\Z$.

The dynamics of a clock with discrete time $\tau\in\Z_n$ can be described by a \emph{shift operator}, a unitary operator $\obs{U}_c$, so that $\obs{U}_c\ket{\tau}_c=\ket{\tau+1\mod n}_c$. As for any unitary operator, one may write $\obs{U}_c=e^{-\frac{i}{\hbar}\obs{H}_c}$ for some $\obs{H}_c$, but the Hamiltonian description will not be needed in the discrete case.

\begin{definition}
\label{def:clock-discrete}
A \emph{discrete ideal clock} is a system with states $\ket{\tau}_c$, $\tau\in\Z_n$, and transitions $\obs{U}_c$, so that
\begin{equation}
\label{eq:evol-discrete-c}
\ket{\tau+1\mod n}_c=\obs{U}_c\ket{\tau}_c.
\end{equation}

We can see the clock as a quantum system with Hilbert space $\hilbert_c$ with orthonormal basis $(\ket{\tau}_c)_{\tau\in\Z_n}$.
The operator $\obs{U}_c$ extends by linearity to a unitary operator on $\hilbert_c$ which will be denoted by $\obs{U}_c$ as well.
We will call the operator $\obs{T}_c:=\sum_{\tau\in\Z}\tau\dyad{\tau}_c$ \emph{time operator}.
\end{definition}

To characterize a finite-time PW system, it is important to notice that the stationarity condition requires that when the clock resets after $n$ ticks, the state of the world system must reset to its initial state as well. If it does not do this, the same clock state $\ket{\tau}_c$ corresponds to multiple world states $\ket{\psi(\tau+kn)}_r$, for all $k\in\Z$, and the PW proposal does not work.
Therefore, stationarity implies, for finite $n$, the cyclicity condition
\begin{equation}
\label{eq:psi-cyclic}
\ket{\psi(\tau+kn)}_r=\ket{\psi(\tau)}_r.
\end{equation}

If this is required for all state vectors from $\hilbert_r$, it leads to the condition that the world's unitary evolution operator is cyclic like the clock's,
\begin{equation}
\label{eq:U-cyclic}
\obs{U}_r^n=\obs{I}_r.
\end{equation}

\begin{definition}
\label{def:PW-system-discrete}
An \emph{ideal discrete Page-Wootters system} consists of a discrete ideal clock and a ``rest of the world'' quantum system with Hilbert space $\hilbert_r$ and discrete time evolution governed by an operator $\obs{U}_r$ on $\hilbert_r$,
\begin{equation}
\label{eq:evol-discrete-r}
\ket{\psi(\tau+1 \mod n)}_r=\obs{U}_r\ket{\psi(\tau)}_r.
\end{equation}

The combined system evolves governed by $\obs{U}=\obs{U}_c\otimes\obs{U}_r$,
\begin{equation}
\label{eq:evol-discrete-cr}
\ket{\tau+1 \mod n}_c\ket{\psi(\tau+1 \mod n)}_r=\obs{U}_c\ket{\tau}_c\obs{U}_r\ket{\psi(\tau)}_r.
\end{equation}

Its \emph{relative temporal states} have the form $\ket{\Psi(\tau)}=\ket{\tau}_c\ket{\psi(\tau)}_r$.
Its allowed \emph{timeless states} have the form
\begin{equation}
\label{eq:PW-state-timeless-discrete}
\kett{\Psi}=\sum_{\tau\in\Z_n}\ket{\tau}_c\ket{\psi(\tau)}_r.
\end{equation}
Due to their invariance to time translations, timeless states satisfy 
\begin{equation}
\label{eq:PW-WDW-discrete}
\obs{U}\kett{\Psi}=\kett{\Psi},
\end{equation}
which is a discrete version of $\obs{H}\kett{\Psi}=0$.
\end{definition}

Based on this mechanism, Page and Wootters answered affirmatively a question they posed in~\cite{PageWootters1983EvolutionWithoutEvolution},
\begin{quote}
whether this sort of evolution (evolution by correlations) always imitates evolution as obtained from the equation of motion.
\end{quote}

Their proposal was developed in various directions~\cite{GiovannettiLloydMaccone2015QuantumTime,LoveridgeMiyadera2019RelativeQuantumTime,MacconeSacha2020QuantumMeasurementsOfTime,CastroruizGiacominiBelenchiaBrukner2020QuantumClocksAndTemporalLocalisabilityOfEventsInThePresenceOfGravitatingQuantumSystems,FotiEtAl2021TimeAndClassicalEqOfMotionFromQEntanglementViaPageWoottersGeneralizedCoherentStates,HohnSmithLock2021TrinityOfRelationalQuantumDynamics,Gambini2022SolutionProblemOfTimeQuantumGravityAlsoTimeArrivalProblemQM,AltaieBeigeHodgson2022TimeAndQuantumClocksAReviewOfRecentDevelopments,Adlam2022WatchingTheClocksInterpretingPageWoottersFormalismAndInternalQRF,Rijavec2022HeisenbergPictureEvolutionWithoutEvolution,SuleymanovCohen2023QuantumFramesOfReferenceAndRelationalFlowOfTime,RidleyCohen2025TwoTimesOrNone}, and even illustrated in experiments~\cite{MorevaBridaGramegnaGiovannettiMacconeGenovese2014TimeFromQuantumEntanglementAnExperimentalIllustration,MorevaGramegnaBridaMacconeGenovese2017QuantumTimeExperimentalMultitimeCorrelations}.

However, the uniqueness of the reconstruction of dynamics from the timeless picture from the PW proposal will be challenged.

\section{Albrecht-Iglesias clock ambiguity}
\label{s:reviews-Albrecht-Iglesias}

Shortly after the PW proposal it became clear that, in the context of quantum gravity, there is a \emph{multiple choice problem}: different choices of ``internal time'' or clock variable can lead, after reduction and quantization, to inequivalent Schr\"odinger equations, and hence different quantum theories~\cite{ADM1962TheDynamicsOfGeneralRelativity,Kuchar1992TimeAndInterpretationsOfQuantumGravity,Isham1993CanonicalQuantumGravityAndTheProblemOfTime,Anderson2012ProblemOfTimeInQuantumGravity,Anderson2017TheProblemOfTime}.
This raises problems regarding the invariance of the method of quantization and the background freedom of quantum gravity.

The possibility of different choices of the internal clock in the PW proposal also leads, more generally, to the clock ambiguity problem~\cite{Albrecht1995TheoryOfEeverythingVsTheoryOfAnything,AlbrechtIglesias2008ClockAmbiguityAndTheEmergenceOfPhysicalLaws,AlbrechtIglesias2012ClockAmbiguityImplicationsNewDevelopments}.
In their proof of the ambiguity, Albrecht and Iglesias assume discrete time, so that $\hilbert_c$ has a finite orthonormal basis $(\ket{\tau_j}_c)_{j}$, but their proof also applies to an infinite countable basis.
Let $(\ket{k}_r)_{k}$ be an orthonormal basis of $\hilbert_r$.
Consider the history explicitly encoded in $\kett{\Psi}=\sum_j \ket{\tau_j}_c\ket{\psi(\tau_j)}_r$, and a randomly chosen history $\ket{\psi'(\tau)}_r$, which may follow a Schr\"odinger evolution based on a randomly chosen Hamiltonian or even be a random walk, encoded in the vector $\kett{\Psi'}=\sum_j \ket{\tau_j}_c\ket{\psi'(\tau_j)}_r$.
Then, there is a unitary transformation $\obs{M}$ of $\hilbert$ so that $\obs{M}\kett{\Psi}=\kett{\Psi'}$.
Let $\kett{\eta_{j,k}}:=\obs{M}^\dagger\ket{\tau_j}_c\ket{k}_r$.
In the orthogonal basis $(\kett{\eta_{j,k}})_{j,k}$ of $\hilbert$, the components of $\kett{\Psi}$ are equal to those of $\kett{\Psi'}$ in the basis $(\ket{\tau_j}_c\ket{k}_r)_{j,k}$,
\begin{equation}
\label{eq:AI-change-basis}
\begin{array}{ll}
\braakett{\eta_{j,k}}{\Psi}
&=\(\bra{\tau_j}\bra{k}\obs{M}\)\kett{\Psi}
=\bra{\tau_j}\bra{k}\(\obs{M}\kett{\Psi}\)\\
&=\bra{\tau_j}\brakett{k}{\Psi'}=\braket{k}{\psi'(\tau_j)}_r.\\
\end{array}
\end{equation}

We now construct another tensor product decomposition $\hilbert=\hilbert_{c'}\otimes\hilbert_{r'}$, determined by the condition that each basis vector $\kett{\eta_{j,k}}$ is a product state $\kett{\eta_{j,k}}=\ket{\tau_j}_{c'}\ket{k}_{r'}$.
By construction, $\kett{\Psi}$ has, in the basis $(\kett{\eta_{j,k}})_{j,k}$, the same components any other vector $\kett{\Psi'}$ has in the original basis $(\ket{\tau_j}_c\ket{k}_r)_{j,k}$, $\kett{\Psi}=\sum_j \ket{\tau_j}_{c'}\ket{\psi'(\tau_j)}_{r'}$.
Since $\kett{\Psi'}=\sum_j \ket{\tau_j}_c\ket{\psi'(\tau_j)}_r$ can be chosen freely, it can encode any history $\ket{\psi'(\tau_j)}_{r}$, and therefore so does $\kett{\Psi}$ in the new tensor product decomposition $\hilbert=\hilbert_{c'}\otimes\hilbert_{r'}$.
The history $\ket{\psi'(\tau_j)}_r$ can be chosen to follow any Hamiltonian $\obs{H}_{r'}$, or even to be a completely random history.
Albrecht concluded in~\cite{Albrecht1995TheoryOfEeverythingVsTheoryOfAnything} that starting with $\kett{\Psi}$
\begin{quote}
one can produce all possible states evolving according to all possible time evolutions simply by choosing a suitable clock.
\end{quote}

\begin{remark}[The pullback-TPS construction.]
\label{rem:pullback-TPS}
The Albrecht-Iglesias proof uses a step that can be understood as a general construction. Let $\hilbert$ be a Hilbert space with a reference tensor product structure
\begin{equation}
\mc{U}:\hilbert_c\otimes\hilbert_r\to\hilbert
\end{equation}
and let $\obs{M}$ be a unitary operator on $\hilbert$. Then $\obs{M}$ can be used passively to define a new tensor product structure by
\begin{equation}
\mc{U}':=\obs{M}^{-1}\mc{U}.
\end{equation}
Equivalently, if $(\ket{\tau_j}_c)_j$ and $(\ket{k}_r)_k$ are orthonormal bases of $\hilbert_c$ and $\hilbert_r$, the product vectors of the new TPS are 
\begin{equation}
\ket{\tau_j}_{c'}\ket{k}_{r'} := \obs{M}^{-1}\left(\ket{\tau_j}_c\ket{k}_r\right).
\end{equation}

Thus, a vector $\kett{\Psi}\in\hilbert$ has, in the primed TPS, the
same components that $\obs{M}\kett{\Psi}$ has in the unprimed TPS:
\begin{equation}
\begin{aligned}
{}_{c'}\bra{\tau_j}{}_{r'}\brakett{k}{\Psi}
&={}_c\bra{\tau_j}{}_r\bra{k}\obs{M}\kett{\Psi}.
\end{aligned}
\end{equation}

This is the sense in which an active unitary transformation of the state
can be reread passively as a change of tensor product structure.
\end{remark}

\section{Two kinds of clock ambiguities}
\label{s:reviews-two-clock-ambiguities}

Albrecht and Iglesias' clock ambiguity establishes that, by changing the TPS, a fixed timeless state can be made to encode arbitrary histories. However, it does not establish the unitary equivalence of the Hamiltonians themselves. 
Let us spell out the difference between the two kinds of ambiguities.

A tensor product structure (TPS) is defined by a unitary isomorphism of Hilbert spaces~\cite{ZanardiLidarLloyd2004QuantumTensorProductStructuresAreObservableInduced}
\begin{equation}
\label{eq:TPS}
\mc{U}:\hilbert_c\otimes\hilbert_r\to\hilbert.
\end{equation}

Another unitary isomorphism
\begin{equation}
\label{eq:TPS-other}
\mc{U}':\hilbert_{c'}\otimes\hilbert_{r'}\to\hilbert
\end{equation}
defines a possibly different TPS.
The TPS given by $\mc{U}'$ is equivalent to the TPS given by $\mc{U}$ if and only if $\mc{U}^{'-1}\mc{U}=\mc{U}_c\otimes\mc{U}_r$ for some unitary maps $\mc{U}_c:\hilbert_c\to\hilbert_{c'}$ and $\mc{U}_r:\hilbert_r\to\hilbert_{r'}$, possibly swapped.
In the following we will assume $\hilbert=\hilbert_c\otimes\hilbert_r$, so $\mc{U}=\obs{I}_{c+r}$.

\begin{definition}
\label{def:ambiguity}
Consider a PW system with $\hilbert=\hilbert_c\otimes\hilbert_r$, $\obs{H}=\obs{H}_c\otimes\obs{I}_r+\obs{I}_c\otimes\obs{H}_r$, and $\kett{\Psi}=\int_\R\ket{\tau}_c\ket{\psi(\tau)}_r\dd \tau$ compatible with the evolution governed by $\obs{H}_c$.
\begin{enumerate}
	\item 
We say that the PW system has the \emph{clock ambiguity problem for histories} if, for any other state vector $\kett{\Psi'}=\int_\R\ket{\tau}_c\ket{\psi'(\tau)}_r\dd \tau$ compatible with the evolution governed by another Hamiltonian $\obs{H}_c'$, there is a TPS $\mc{U}':\hilbert_{c'}\otimes\hilbert_{r'}\to\hilbert$ so that
\begin{equation}
\label{eq:clock-ambiguity-histories}
\kett{\Psi}=\mc{U}'\int_\R\ket{\tau}_{c'}\ket{\psi'(\tau)}_{r'}\dd \tau.
\end{equation}
	\item 
We say that the PW system has the \emph{clock ambiguity problem for histories and laws} if it has the clock ambiguity problem for histories and, in addition,
\begin{equation}
\label{eq:clock-ambiguity-laws}
\obs{H}_c\otimes\obs{I}_r+\obs{I}_c\otimes\obs{H}_r=\mc{U}'\(\obs{H}_{c'}\otimes\obs{I}_{r'}+\obs{I}_{c'}\otimes\obs{H}_{r'}\)\mc{U}^{'-1}.
\end{equation}
\end{enumerate}
\end{definition}

An equivalent way to express the two versions of the clock ambiguity from Definition~\ref{def:ambiguity} is as unitary equivalence of any two PW systems with $\dim\hilbert_c=\dim\hilbert_{c'}$ and $\dim\hilbert_r=\dim\hilbert_{r'}$. Different choices of the clock subsystem are related by a unitary transformation, which changes the TPS $\mc{U}:\hilbert_c\otimes\hilbert_r\to\hilbert$ into a possibly different TPS $\mc{U}':\hilbert_{c'}\otimes\hilbert_{r'}\to\hilbert$.
Then, the map
\begin{equation}
\label{eq:ambiguity-equivalent}
\begin{array}{l}
\obs{S}:\hilbert_c\otimes\hilbert_r\to\hilbert_{c'}\otimes\hilbert_{r'}\\
\obs{S}=\mc{U}^{'-1} \mc{U}\\
\end{array}
\end{equation}
is a unitary equivalence between the two PW systems $\hilbert_c\otimes\hilbert_r$ and $\hilbert_{c'}\otimes\hilbert_{r'}$.
If such a map $\obs{S}$ exists mapping any history in the first PW system into a history in the second PW system, the two PW systems have the clock ambiguity problem for histories. If $\obs{S}$ also satisfies condition~\eqref{eq:clock-ambiguity-laws}, they have the clock ambiguity problem for histories and laws.

With Definition~\ref{def:ambiguity}, Albrecht's result~\cite{Albrecht1995TheoryOfEeverythingVsTheoryOfAnything}, proved in~\sref{s:reviews-Albrecht-Iglesias}, becomes
\begin{proposition}
\label{thm:Albrecht-Iglesias-clock-ambiguity}
Any finite-dimensional PW system has the clock ambiguity problem for histories.
\end{proposition}

The ambiguity proved by Albrecht and Iglesias concerns histories, but not evolution laws: their proof establishes condition~\eqref{eq:clock-ambiguity-histories}, not condition~\eqref{eq:clock-ambiguity-laws}.

At first sight, condition~\eqref{eq:clock-ambiguity-laws} may seem impossible to satisfy for any pair of Hamiltonians $\obs{H}_r$ and $\obs{H}_r'$, because their spectra may be different.
But we will see that any difference in the spectra is washed out by the spectrum of $\obs{H}_c$.

We will see that, for ideal clocks with finite, discrete, or continuous  time, the stronger version of the clock ambiguity holds. Even if we require that there is no interaction between the clock and the rest of the world, any ideal clock PW system has the clock ambiguity problem for both histories and laws. This ambiguity is maximal, in the sense that the clock's Hamiltonian removes any difference between any two Hamiltonians $\obs{H}_r$ and $\obs{H}_{r'}$ and any two histories $\ket{\psi(\tau)}_r$ and $\ket{\psi'(\tau)}_r$.

\section{Maximal clock ambiguity for histories and laws}
\label{s:ambiguity}

I now prove the stronger ambiguity for ideal clocks, which extends Albrecht and Iglesias' ambiguity so that
\begin{enumerate}
	\item 
it includes, besides unitary equivalences of histories, unitary equivalences of the Hamiltonians,
	\item 
	the ambiguity is maximal even under Marletto and Vedral's assumption that the clock and the rest of the world do not interact.
\end{enumerate}

For this, I assume all Hilbert spaces separable.
For clarity, I will first prove the discrete time case, treating in the Appendix the more technical cases.

\subsection{Maximal ambiguity for finite time}
\label{s:ambiguity-discrete-time-finite}

Let us start with the simplest example.
\begin{example}
\label{example:two-qubits}
Consider the two-qubit PW system:
\begin{equation}
\label{eq:qubit-example-U}
\begin{cases}
\obs{U}=\obs{U}_c\otimes\obs{U}_r=\sigma_x\otimes\sigma_x, \\
\obs{T}_c=\dyad{1}_c=(\obs{I}_2-\sigma_z)/2, \\
\kett{\Psi}=\frac{1}{\sqrt{2}}\(\ket{0}_c\ket{1}_r+\ket{1}_c\ket{0}_r\).
\end{cases}
\end{equation}
Here $(\ket{0},\ket{1})$ is the $\sigma_z$ eigenbasis. In this basis, the operator $\sigma_z$ is the clock operator and $\sigma_x$ is the
shift operator.

Consider another two-qubit discrete PW system:
\begin{equation}
\label{eq:qubit-example-U-b}
\begin{cases}
\obs{U}'=\obs{U}_{c'}\otimes\obs{U}_{r'}=\sigma_x\otimes\obs{I}_2, \\
\obs{T}_{c'}=\dyad{1}_{c'}=(\obs{I}_2-\sigma_z)/2, \\
\kett{\Psi'}=\frac{1}{\sqrt{2}}(\ket{0}_{c'}\ket{1}_{r'}+\ket{1}_{c'}\ket{1}_{r'}).
\end{cases}
\end{equation}

The two PW systems follow distinct dynamical laws for the rest of the world. However, they are unitarily equivalent as PW systems with histories and laws.
\end{example}
\begin{proof}
Define $\obs{S}:\hilbert_c\otimes\hilbert_r\to\hilbert_{c'}\otimes\hilbert_{r'}$ on the product basis so that it maps two pairs of histories into one another,
\begin{equation}
\label{eq:qubit-example-S-basis}
\begin{aligned}
\obs{S}\ket{0}_c\ket{1}_r & = \ket{0}_{c'}\ket{1}_{r'}, \\
\obs{S}\ket{1}_c\ket{0}_r & = \ket{1}_{c'}\ket{1}_{r'}, \\
\obs{S}\ket{0}_c\ket{0}_r & = \ket{0}_{c'}\ket{0}_{r'}, \\
\obs{S}\ket{1}_c\ket{1}_r & = \ket{1}_{c'}\ket{0}_{r'}.
\end{aligned}
\end{equation}

This gives
\begin{equation}
\label{eq:qubit-example-S}
\begin{aligned}
\obs{S}
=&\ket{0}_{c'}\bra{0}_c\otimes
  \(\ket{0}_{r'}\bra{0}_r+\ket{1}_{r'}\bra{1}_r\) \\
+&\ket{1}_{c'}\bra{1}_c\otimes
  \(\ket{0}_{r'}\bra{1}_r+\ket{1}_{r'}\bra{0}_r\).
\end{aligned}
\end{equation}

This has the form of a controlled-not (or, more generally, a controlled-shift) on the rest of the world. It is a unitary operator, and its inverse has the same form,
\begin{equation}
\label{eq:qubit-example-S-inv}
\begin{aligned}
\obs{S}^{-1}
=&\ket{0}_{c}\bra{0}_{c'}\otimes
  \(\ket{0}_{r}\bra{0}_{r'}+\ket{1}_{r}\bra{1}_{r'}\) \\
+&\ket{1}_{c}\bra{1}_{c'}\otimes
  \(\ket{0}_{r}\bra{1}_{r'}+\ket{1}_{r}\bra{0}_{r'}\).
\end{aligned}
\end{equation}

A direct calculation gives
\begin{equation}
\label{eq:qubit-example-S-U}
\obs{S}(\sigma_x\otimes\sigma_x)\obs{S}^{-1}=\sigma_x\otimes\obs{I}_2.
\end{equation}

By construction, from equation~\eqref{eq:qubit-example-S-basis},
\begin{equation}
\label{eq:qubit-example-state-map}
\obs{S}\kett{\Psi}=\kett{\Psi'}.
\end{equation}

Therefore, the histories and the laws are unitarily equivalent, even though 
$\obs{U}_r=\sigma_x$ and $\obs{U}_{r'}=\obs{I}_2$ are distinct and have different spectra.

Thus both evolution laws are noninteracting, but the ambiguity remains.
\end{proof}

Example~\ref{example:two-qubits} is a result about finite cyclic PW systems, whose primitive dynamical datum is the step operator $\obs{U}$, not a preferred logarithm $\obs{H}$.
A Hamiltonian logarithm of $\obs{U}$ may be introduced, but it is extra branch data not fixed by the finite cyclic PW system.
A discrete PW system does not have access to a continuum of instantaneous states, so the discrete time unitary operator $\obs{U}$ is the suitable formalism.

Note that the Hamiltonians
\begin{equation}
\label{eq:qubit-example-H}
\obs{H}=\sigma_x\otimes\obs{I}_2+\obs{I}_2\otimes\sigma_x
\end{equation}
and
\begin{equation}
\label{eq:qubit-example-Hb}
\obs{H}'=\sigma_x\otimes\obs{I}_2
\end{equation}
are, obviously, not unitarily equivalent. They have different spectra,
\begin{equation}
\label{eq:qubit-example-H-spectra}
\begin{array}{ll}
\varsigma(\obs{H})&=\(-2,0,0,2\)\quad\text{and} \\
\varsigma(\obs{H}')&=\(-1,-1,1,1\).\\
\end{array}
\end{equation}

But the spectra of the unitary operators $\obs{U}$ and $\obs{U}'$ are identical,
\begin{equation}
\label{eq:qubit-example-U-spectra}
\begin{array}{ll}
\varsigma(\obs{U})&=\(-1,-1,1,1\)\quad\text{and} \\
\varsigma(\obs{U}')&=\(-1,-1,1,1\).\\
\end{array}
\end{equation}

This may seem like a discrepancy: if $\obs{S}$ transforms $\obs{U}$ into $\obs{U}'$, does this mean that it must transform $\obs{H}$ into $\obs{H}'$?
But this is impossible, because $\obs{H}$ and $\obs{H}'$ have different spectra. What happens here is that $\obs{S}$ transforms $\obs{H}$ into a Hamiltonian generator of $\obs{U}'$ other than $\obs{H}'$.
The reason is that a unitary operator has infinitely many logarithms.

\begin{remark}
\label{rem:UvH}
The fact that there are more Hamiltonian generators for the same one-step unitary evolution operator $\obs{U}$ may give the impression that, by dropping the Hamiltonian and using only $\obs{U}$, we lose essential information, which is the cause of the ambiguity from Example~\ref{example:two-qubits}.
But this information that singles out a single Hamiltonian generator for $\obs{U}$ was not there to begin with: the history of the discrete PW system has only two states, connected by $\obs{U}$, not a continuum of states connected infinitesimally by the Schr\"odinger equation for $\obs{H}$.
\end{remark}

This may seem like an exceptional situation, but it is the general rule.
The reason is that the spectrum of $\obs{U}_{c}$ is $\varsigma\(\obs{U}_c\)=\{-1,1\}$, and since $\obs{U}_r^2=\obs{I}_2$, its spectrum must be included in $\{-1,1\}$. This is the case for both $\obs{U}_r=\sigma_x$ and $\obs{U}_{r'}=\obs{I}_2$. When combining their spectra, we get the spectra from equation~\eqref{eq:qubit-example-U-spectra} regardless of $\obs{U}_{r'}$, which makes the existence of the transformation $\obs{S}$ possible.

In the continuous case, the same happens, but for Hamiltonians too, as we will see in~\sref{s:ambiguity-continuous-time}.

For finite time with $n$ ticks, let the clock Hilbert space be
$\hilbert_c\cong\C^n$, with the standard time basis
$(\ket{\tau}_c)_{\tau\in\Z_n}$. The standard clock and shift
operators are
\begin{equation}
\label{eq:finite-clock-shift}
\begin{cases}
\obs{Z}_n\ket{\tau}_c=\omega_n^\tau\ket{\tau}_c,\\
\obs{X}_n\ket{\tau}_c=\ket{\tau+1}_c, \\
\end{cases}
\end{equation}
where $\omega_n=e^{2\pi i/n}$. They satisfy
\begin{equation}
\label{eq:finite-clock-shift-rel}
\obs{X}_n^n=\obs{Z}_n^n=\obs{I}_c,
\qquad
\obs{Z}_n\obs{X}_n=\omega_n\obs{X}_n\obs{Z}_n.
\end{equation}

The time operator is
\begin{equation}
\label{eq:finite-time-operator}
\obs{T}_c=\sum_{\tau\in\Z_n}\tau\dyad{\tau}_c.
\end{equation}

The clock transition is governed by the shift operator $\obs{U}_c=\obs{X}_n$.

The clock and shift operators are generalizations of the Pauli matrices, so for $n=2$ they are
\begin{equation}
\label{eq:finite-qubit-clock-shift}
\obs{Z}_2=\sigma_z,
\qquad
\obs{X}_2=\sigma_x.
\end{equation}

The total unitary evolution of the finite cyclic PW system has the form
\begin{equation}
\label{eq:evol-discrete-finite}
\obs{U}=\obs{X}_n\otimes\obs{U}_r.
\end{equation}

Due to the stationarity condition for the discrete case, $\obs{U}\kett{\Psi}=\kett{\Psi}$, both the clock and the world return to their initial states after $\tau=n$ ticks. Therefore,
\begin{equation}
\label{eq:evol-discrete-loop}
\obs{U}_r^n=\obs{I}_r
\qquad
\text{and}
\qquad
\obs{U}^n=\obs{I}_{c+r}.
\end{equation}

Example~\ref{example:two-qubits} generalizes easily to any finite clock dimension $n$.

\begin{theorem}[Maximal ambiguity, finite time]
\label{thm:maximal-ambiguity-finite-cyclic}
Let two finite PW systems have the same clock dimension $n$ and
$\dim\hilbert_r=\dim\hilbert_{r'}$. Let their one-step evolution operators
be
\begin{equation}
\label{eq:finite-cyclic-two-systems}
\obs{U}=\obs{X}_n\otimes\obs{U}_r,
\qquad
\obs{U}'=\obs{X}_n'\otimes\obs{U}_{r'},
\end{equation}
where $\obs{U}_r^n=\obs{I}_r$ and $\obs{U}_{r'}^n=\obs{I}_{r'}$. Then, for
any two temporal states $\ket{\Psi(0)}=\ket{0}_c\ket{\psi(0)}_r$ and
$\ket{\Psi'(0)}=\ket{0}_{c'}\ket{\psi'(0)}_{r'}$, there is a unitary operator
$\obs{S}:\hilbert_c\otimes\hilbert_r\to\hilbert_{c'}\otimes\hilbert_{r'}$
such that, for all $\tau\in\Z_n$,
\begin{equation}
\label{eq:PW-unitarily-equivalent-state-finite}
\ket{\Psi'(\tau)}=\obs{S}\ket{\Psi(\tau)},
\end{equation}
and
\begin{equation}
\label{eq:PW-unitarily-equivalent-U-finite}
\obs{U}'=\obs{S}\obs{U}\obs{S}^{-1}.
\end{equation}
\end{theorem}
\begin{proof}
Choose two orthonormal bases, $(\ket{e_k}_r)_k$ of $\hilbert_r$, and
$(\ket{v_k}_{r'})_k$ of $\hilbert_{r'}$, so that $\ket{e_1}_r=\ket{\psi(0)}_r$ and $\ket{v_1}_{r'}=\ket{\psi'(0)}_{r'}$.
For $\tau\in\Z_n$ define
\begin{equation}
\label{eq:finite-cyclic-basis}
\begin{array}{ll}
\ket{\tau,k} &:=\ket{\tau}_c\obs{U}_r^\tau\ket{e_k}_r,\quad\text{and}\\
\ket{\tau,k}' &:=\ket{\tau}_{c'}\obs{U}_{r'}^\tau\ket{v_k}_{r'}.\\
\end{array}
\end{equation}
Because the clock states $\ket{\tau}_c$ are mutually orthogonal, both
families in equation~\eqref{eq:finite-cyclic-basis} are orthonormal bases
of their total Hilbert spaces. Define $\obs{S}$ by
\begin{equation}
\label{eq:finite-cyclic-S}
\obs{S}\ket{\tau,k}=\ket{\tau,k}'.
\end{equation}
Then $\obs{S}$ is unitary. Moreover,
\begin{equation}
\label{eq:finite-cyclic-intertwining}
\begin{aligned}
\obs{S}\obs{U}\ket{\tau,k}
&=\obs{S}\ket{\tau+1 \mod n,k} \\
&=\ket{\tau+1 \mod n,k}' \\
&=\obs{U}'\ket{\tau,k}' \\
&=\obs{U}'\obs{S}\ket{\tau,k},
\end{aligned}
\end{equation}
so $\obs{S}\obs{U}=\obs{U}'\obs{S}$, which is equivalent to
\eqref{eq:PW-unitarily-equivalent-U-finite}. For $k=1$, equation
\eqref{eq:finite-cyclic-S} gives
\begin{equation}
\obs{S}\ket{\tau}_c\obs{U}_r^\tau\ket{\psi(0)}_r
=\ket{\tau}_{c'}\obs{U}_{r'}^\tau\ket{\psi'(0)}_{r'},
\end{equation}
which is exactly equation~\eqref{eq:PW-unitarily-equivalent-state-finite}.
\end{proof}

The proof of Theorem~\ref{thm:maximal-ambiguity-finite-cyclic} also suggests that we can take $n=\infty$, but we need to replace $\obs{H}_c=X_n$ with the bilateral shift on $\ell^2(\Z)$, whose spectrum is the full unit circle $S^1=e^{-\frac{i}{\hbar}\R}$ instead of the $n$ $n$-th roots of unity $\omega_n=e^{2\pi i/n}$. The same multiplicative spectral-washing argument applies. We prove this case in Appendix~\ref{s:ambiguity-discrete-time-infinite}.

\subsection{Maximal ambiguity for continuous time}
\label{s:ambiguity-continuous-time}

The continuous case $\tau\in\R$ as in Definition~\ref{def:PW-system} is similar to the discrete case, but the spectral-washing argument is applied to the Hamiltonians. Since $\obs{H}_c=-i\hbar\partial_\tau$, $\varsigma(\obs{H}_c)=\R$, and one obtains
\begin{equation}
\label{eq:spectral-eating-H}
    \varsigma(\obs{H})=\varsigma(\obs{H}_c)+\varsigma(\obs{H}_r)=\R+\varsigma(\obs{H}_r)=\R.
\end{equation}

\begin{proposition}
\label{thm:continuous-zero}
For any PW system as in Definition~\ref{def:PW-system}, there is another TPS $\mc{U}':\hilbert_{c'}\otimes\hilbert_{r'}\to\hilbert$ in which the PW system is expressed as $\obs{H}=\obs{H}_{c'}\otimes\obs{I}_{r'}$, that is, $\obs{H}_{r'}=0$.
\end{proposition}
\begin{proof}
Let $\mc{U}:\hilbert_c\otimes\hilbert_r\to\hilbert$ be the original TPS. Define
\begin{equation}
\label{eq:PW-trivial-rest-S}
\obs{S}=e^{\frac{i}{\hbar}\obs{T}_c\otimes\obs{H}_r}.
\end{equation}
Using $[\obs{T}_c,\obs{H}_c]=i\hbar$, we get
\begin{equation}
\label{eq:PW-trivial-rest-Hc-conjugated}
\obs{S}(\obs{H}_c\otimes\obs{I}_r)\obs{S}^{-1}
=\obs{H}_c\otimes\obs{I}_r-\obs{I}_c\otimes\obs{H}_r,
\end{equation}
and
\begin{equation}
\label{eq:PW-trivial-rest-Hr-conjugated}
\obs{S}(\obs{I}_c\otimes\obs{H}_r)\obs{S}^{-1}
=\obs{I}_c\otimes\obs{H}_r.
\end{equation}
Therefore
\begin{equation}
\label{eq:PW-trivial-rest-H-conjugated}
\obs{S}\obs{H}\obs{S}^{-1}=\obs{H}_c\otimes\obs{I}_r.
\end{equation}

Now apply the pullback-TPS construction (see Remark~\ref{rem:pullback-TPS}) to $\obs{S}$ to define a new TPS by
\begin{equation}
\label{eq:PW-trivial-rest-new-TPS-def}
\mc{U}':=\mc{U}\circ \obs{S}^{-1}.
\end{equation}

Equivalently, define the primed TPS by its product basis vectors
\begin{equation}
\label{eq:PW-trivial-rest-new-products}
\mc{U}'\ket{\tau}_{c'}\ket{\chi}_{r'}
=
\mc{U}\obs{S}^{-1}\ket{\tau}_c\ket{\chi}_r.
\end{equation}

Then, the pullback-TPS construction results in the expression of $\obs{H}$ in the primed TPS as $\obs{S}\obs{H}\obs{S}^{-1}$. Equation~\eqref{eq:PW-trivial-rest-H-conjugated} therefore says precisely that, in the primed TPS,
\begin{equation}
\label{eq:PW-trivial-rest-H-new-TPS}
\obs{H}=\obs{H}_{c'}\otimes\obs{I}_{r'} .
\end{equation}

Therefore, $\obs{H}_{r'}=0$.
\end{proof}

\begin{remark}
\label{rem:example:continuous}
From Proposition~\ref{thm:continuous-zero} follows immediately that any PW system is equivalent to any other PW system on spaces of equal dimension.
\end{remark}

\begin{remark}
\label{rem:interaction-picture-not}
The unitary $\obs{S}$ can also be seen, in the original TPS, as a clock-controlled change of picture. But in the proof it is used passively through the pullback-TPS construction from Remark~\ref{rem:pullback-TPS}, not to change the picture. The factor $\hilbert_{r'}$ is not the original factor $\hilbert_r$ with a different notation and in the Heisenberg picture: its product vectors are transformed nonlocally by $\obs{S}^{-1}$. Just like when the pullback-TPS construction was applied in Albrecht and Iglesias' proof, and could be interpreted as an ordinary change of basis in the same Hilbert space, the same calculation can look like an interaction-picture transformation in the old TPS, but in the present argument it defines a different clock/rest decomposition.
\end{remark}

\begin{remark}
\label{rem:interaction-picture-maximal-ambiguity}
Let us entertain the possibility that the stronger clock ambiguity from Proposition~\ref{thm:continuous-zero} is merely the same dynamics expressed in a different picture and in a different TPS. Even then, the bare PW structure would not tell us which representation is the fundamental Schr\"odinger picture. Is it the one associated with the original TPS, the one associated with the primed TPS, or another one with a different world Hamiltonian? If the answer is not fixed by additional physical input, then any two Hamiltonians become indistinguishable in a bare PW system.
\end{remark}

Let us now give the continuous-time proof, first constructively and then in terms of spectral uniformization.

\begin{theorem}[Maximal ambiguity]
\label{thm:maximal-ambiguity}
Any two ideal PW systems with $\dim\hilbert_r=\dim\hilbert_{r'}$ are unitarily equivalent, in the sense that, for any temporal states $\ket{\Psi(0)}\in\hilbert$ and $\ket{\Psi'(0)}\in\hilbert'$ and all $\tau$, there is a unitary operator $\obs{S}:\hilbert_c\otimes\hilbert_r\to\hilbert_{c'}\otimes\hilbert_{r'}$ so that
\begin{equation}
\label{eq:PW-unitarily-equivalent-H}
\obs{H}'=\obs{S}\obs{H}\obs{S}^{-1}
\end{equation}
and, for all $\tau\in\R$,
\begin{equation}
\label{eq:PW-unitarily-equivalent-state}
\ket{\Psi'(\tau)}=\obs{S}\ket{\Psi(\tau)}.
\end{equation}
\end{theorem}

\begin{proof}
This proof shows the unitary equivalence more explicitly, and also works for $\tau$-dependent world Hamiltonians $\obs{H}_r(\tau)$.
To simplify the notation, I assume that $\obs{H}_r$ has discrete spectrum, but this will not affect the result.
Choose an orthonormal basis in $\hilbert_r$, 
\begin{equation}
\label{eq:PW-basis}
\(\ket{e_1}_r, \ket{e_2}_r,\ldots\).
\end{equation}

For each index $k$ of the basis vectors~\eqref{eq:PW-basis}, denote 
\begin{equation}
\label{eq:PW-tk-translation}
\ket{\tau,k}:=\obs{U}(\tau,0)\ket{0}_c\ket{e_k}_r=\ket{\tau}_c \obs{U}_r(\tau,0)\ket{e_k}_r,
\end{equation}
where $\obs{U}(\tau,0)$ and $\obs{U}_r(\tau,0)$ are obtained by time-ordering (anti-ordering for $\tau<0$) from the time-dependent Hamiltonians.
Then,
\begin{equation}
\label{eq:PW-tk-ortho}
\braket{\tau,k}{\tau',k'}\stackrel{\eqref{eq:t-basis}}{=}\delta_{kk'}\delta(\tau-\tau').
\end{equation}

Let $\hilbert_k$ be the Hilbert space spanned by the vectors $\ket{\tau,k}$, for all $\tau\in\R$ and fixed $k$. Then, from equations~\eqref{eq:PW-tk-translation} and~\eqref{eq:PW-tk-ortho} we get that 
the restriction of the total time evolution operator $\obs{U}(\tau,0)$ to $\hilbert_k$ coincides with the translation with $\tau$. Therefore, due to Stone's theorem \cite{Stone1932OnOneParameterUnitaryGroupsInHilbertSpace} (\cite{Hall2013QuantumTheoryForMathematicians}, Theorem 10.15 p. 210), the restriction of the total Hamiltonian $\obs{H}$ to $\hilbert_k$ coincides with the operator $-i\hbar\pdv{\tau}|_{\hilbert_k}$,
\begin{equation}
\label{eq:subspace-k-translation}
\obs{H}|_{\hilbert_k}=-i\hbar\pdv{\tau}|_{\hilbert_k}.
\end{equation}

We define the unitary equivalence $\obs{S}:\hilbert\to\hilbert'$ by
\begin{equation}
\label{eq:unitary-equiv-trivial}
\obs{S}\ket{\tau,k}=\ket{\tau}_{c'}\ket{v_k}_{r'},
\end{equation}
where $\(\ket{v_k}_{r'}\)_k$ is an orthonormal basis of $\hilbert_{r'}$.

Since $\hilbert=\bigoplus_k\hilbert_k$, the total Hamiltonian $\obs{H}$ is unitarily equivalent to the Hamiltonian $\obs{H}'$ on $\hilbert'$ where
\begin{equation}
\label{eq:total-space-translation}
\obs{H}'=-i\hbar\pdv{\tau}\otimes \obs{I}_{r'}.
\end{equation}

Moreover, $\obs{S}$ maps each history $(\obs{U}(\tau,0)\ket{0}_c\ket{e_k}_r)_{\tau\in\R}$ of the first PW system into a history $(e^{-\frac{i}{\hbar}\obs{H}_{c'} \tau}\ket{0}_{c'}\ket{v_k}_{r'})_{\tau\in\R}$ of the trivial PW system, as seen in equation~\eqref{eq:unitary-equiv-trivial}.
Therefore, our PW system is unitarily equivalent to a trivial PW system.

Now consider two general PW systems, not necessarily trivial, for which $\dim\hilbert_{r'}=\dim\hilbert_r$. Since both PW systems are unitarily equivalent to the trivial PW system, they are unitarily equivalent to one another.
\end{proof}

Let us also give a proof discussing explicitly the spectral measure class.
\begin{proof}[Alternative proof]
At first sight, this may seem impossible, because the existence of a unitary transformation $\obs{S}$ seems blocked by the fact that $\obs{H}_r$ and $\obs{H}_{r'}$ may have different spectra.
But $\obs{H}$ and $\obs{H}'$ have the same spectrum, as we shall see.
However, the proof is less direct than in the discrete case.

We will first prove the equivalence of a PW system with a \emph{trivial PW system}, \ie a PW system for which the Hamiltonian of the rest of the world is $\obs{H}_{r'}=0$.

Consider a PW system and a trivial PW system,
\begin{equation}
\label{eq:PW-maximal-ambiguity}
\begin{cases}
\hilbert=\hilbert_c\otimes\hilbert_r \\
\obs{H}=\obs{H}_c\otimes\obs{I}_r+\obs{I}_c\otimes\obs{H}_r\\
\kett{\Psi}=\int_\R\ket{\tau}_c\ket{\psi(\tau)}_r\dd \tau \\
\end{cases}
\hspace{-0.15in}
\begin{cases}
\hilbert'=\hilbert_{c'}\otimes\hilbert_{r'} \\
\obs{H}'=\obs{H}_{c'}\otimes\obs{I}_{r'}\\
\kett{\Psi'}=\int_\R\ket{\tau}_{c'}\ket{v}_{r'}\dd \tau \\
\end{cases}
\end{equation}
so that $\dim\hilbert_r=\dim\hilbert_{r'}$ and $\obs{H}_{r'}=0$.

The spectrum of $\obs{H}$ is the Minkowski sum of the spectra of $\obs{H}_c$ and $\obs{H}_r$, $\varsigma(\obs{H})=\varsigma(\obs{H}_c)+\varsigma(\obs{H}_r)$. But since the spectrum $\varsigma(\obs{H}_c)=\R$, this implies that $\varsigma(\obs{H})=\R+\varsigma(\obs{H}_r)=\R$.
Therefore, regardless of the Hamiltonian $\obs{H}_r$, the total Hamiltonian $\obs{H}$ has the spectrum $\R$. Let $(\ket{\varepsilon}_c)_{\varepsilon\in\R}$ be an eigenbasis of $\obs{H}_c$. Heuristically, each of the eigenspaces $\hilbert_{\varepsilon}\subset\hilbert$ of $\obs{H}$, spanned by tensor products of the form $\ket{\varepsilon-a}_c\ket{a,m}_r$, where $\ket{a,m}_r$ is an eigenvector corresponding to the eigenvalue $a$ of $\obs{H}_r$ and $m$ distinguishes eigenvectors of the same multiplicity of $a$, has dimension $\dim\hilbert_{\varepsilon}=\dim\hilbert_r$.

This heuristic should be made rigorous.
Unlike the discrete case, for $\tau\in\R$ this is not enough to prove the unitary equivalence, we also need to show that the spectral measure class of $\obs{H}$ is independent of the Hamiltonian operator $\obs{H}_r$.
But this is ensured because the spectral measure of $\obs{H}_c=-i\hbar\pdv{\tau}$ is the Lebesgue measure, and when combining $\obs{H}_c$ with $\obs{H}_r$ in $\obs{H}_c\otimes\obs{I}_r+\obs{I}_c\otimes\obs{H}_r$, this measure is only translated with increments equal to the eigenvalues of $\obs{H}_r$. Since the Lebesgue measure is invariant under translations, the eigenvalues of $\obs{H}_r$ are literally ``lost in translation'' and could as well all be just $0$.
This applies even when the spectrum of $\obs{H}_r$ has continuous parts.
To see this, we'll work in an eigenbasis that diagonalizes both $\obs{H}_c$ and $\obs{H}_r$. In the eigenbasis $(\ket{\varepsilon}_c)_{\varepsilon\in\R}$, $\obs{H}_c$ acts by multiplication by the clock's energy eigenvalues $\varepsilon$.
According to the spectral theorem~\cite{ReedSimon1980MethodsOfModernMathematicalPhysics,Hall2013QuantumTheoryForMathematicians} applied to $\obs{H}_r$, there is a standard measure space $(X,\mu)$ and a measurable multiplicity function $m:X\to\N$ so that
\begin{equation}
\label{eq:spectral-theorem-space-r}
\hilbert_r\cong\int_X^\otimes\C^{m(\xi)}\dd\mu(\xi)
\end{equation}
and a function $\alpha:X\to\R$ so that, in this representation, $\obs{H}_r$ is diagonal and acts by multiplication by $\alpha(\xi)$
\begin{equation}
\label{eq:spectral-theorem-obs}
\matrixel{\xi}{\obs{H}_r}{\psi}_r=\alpha(\xi)\braket{\xi}{\psi}_r.
\end{equation}

In terms of square integrable functions valued in fibers $\hilbert_{\xi}:=\C^{m(\xi)}$, $\hilbert$ becomes~\cite{ReedSimon1980MethodsOfModernMathematicalPhysics,Hall2013QuantumTheoryForMathematicians}
\begin{equation}
\label{eq:spectral-theorem-space-total}
\begin{aligned}
\hilbert
&\cong L^2(\R_{E},\dd\varepsilon)_c\otimes\hilbert_r \\
&\cong L^2\(\R\times X,\dd\varepsilon\dd\mu(\xi);\C^{m(\xi)}\).
\end{aligned}
\end{equation}

Since we diagonalized both $\obs{H}_c$ and $\obs{H}_r$, $\obs{H}$ acts by multiplication by $\varepsilon+\alpha(\xi)$, so
\begin{equation}
\label{eq:spectral-theorem-hamiltonian-total}
\braa{\varepsilon,\xi}\obs{H}\kett{\Psi}=\(\varepsilon+\alpha(\xi)\)\braakett{\varepsilon,\xi}{\Psi}.
\end{equation}

To construct a unitary transformation $\obs{S}$ connecting the two PW systems, for each $\xi\in X$, define $\obs{S}_{\xi}$ on $\hilbert_{\xi}$ by
\begin{equation}
\label{eq:spectral-U-xi}
\braa{\varepsilon,\xi}\obs{S}_{\xi}\kett{\Psi}=\braakett{\varepsilon-\alpha(\xi),\xi}{\Psi}.
\end{equation}

This transformation acts like a translation $\braakett{\varepsilon,\xi}{\Psi}\mapsto\braakett{\varepsilon-\alpha(\xi),\xi}{\Psi}$ on $L^2(\R_{E},\dd\varepsilon)_c$, so $\obs{S}_{\xi}$ is unitary.
Since $\dd\varepsilon$ is the Lebesgue measure, it is invariant under translations, so integrating over $X$ preserves the norm for any normalized $\kett{\Psi}\in\hilbert$,
$\int_X\int_{\R}\abs{\braakett{\varepsilon-\alpha(\xi),\xi}{\Psi}}^2\dd\varepsilon\dd\mu(\xi) 
= \int_X\int_{\R}\abs{\braakett{\varepsilon,\xi}{\Psi}}^2\dd\varepsilon\dd\mu(\xi)$.
Therefore, the operator $\obs{S}$ on $\hilbert$ obtained by direct integration over $X$ of all operators $\obs{S}_{\xi}$ defined on $\hilbert_{\xi}$ is unitary.
Then,
\begin{equation}
\label{eq:spectral-U-H}
\begin{aligned}
\braa{\varepsilon,\xi}\obs{S}_{\xi}\obs{H}\kett{\Psi}
&\stackrel{\eqref{eq:spectral-U-xi}}{=}\braa{\varepsilon-\alpha(\xi),\xi}\obs{H}\kett{\Psi} \\
&\stackrel{\eqref{eq:spectral-theorem-hamiltonian-total}}{=}\(\(\varepsilon-\alpha(\xi)\)+\alpha(\xi)\)\braakett{\varepsilon-\alpha(\xi),\xi}{\Psi} \\
&\stackrel{\eqref{eq:spectral-U-xi}}{=}\varepsilon\braa{\varepsilon,\xi}\obs{S}\kett{\Psi}.
\end{aligned}
\end{equation}

Thus $\obs{S}\obs{H}=\obs{H}_0\obs{S}$, where $\obs{H}_0$ acts by multiplication by $\varepsilon$. Since $\obs{S}$ is unitary, this is equivalent to $\obs{S}\obs{H}\obs{S}^{-1}=\obs{H}_0$.
Therefore, the transformed Hamiltonian $\obs{S}\obs{H}\obs{S}^{-1}$ acts, in the representation~\eqref{eq:spectral-theorem-space-total}, as multiplication by $\varepsilon$.
This action is independent of the operator $\obs{H}_r$, whose spectrum was completely washed out under this transformation.
This proves that the Hamiltonian $\obs{H}_c\otimes\obs{I}_r+\obs{I}_c\otimes\obs{H}_r$ is unitarily equivalent to the Hamiltonian $\obs{H}_c\otimes\obs{I}_r$.

At this point, the transformation $\obs{S}$ establishes the unitary equivalence of the PW system with a trivial one. By composing $\obs{S}$, if necessary, with a local unitary transformation on the Hilbert space $\hilbert_{r'}$, one can find a unitary transformation that also satisfies $\ket{\Psi'(\tau)}=\obs{S}\ket{\Psi(\tau)}$.
More explicitly, choose unitary maps $\obs{W}_c:\hilbert_c\to\hilbert_{c'}$ and $\obs{W}_r:\hilbert_r\to\hilbert_{r'}$ such that $\obs{W}_c\ket{\tau}_c=\ket{\tau}_{c'}$ and $\obs{W}_r\ket{\psi(0)}_r=\ket{\psi'(0)}_{r'}$. Then define $\obs{S}=e^{-\frac{i}{\hbar}\obs{T}_{c'}\otimes\obs{H}_{r'}}(\obs{W}_c\otimes\obs{W}_r)e^{\frac{i}{\hbar}\obs{T}_c\otimes\obs{H}_r}$.
Indeed,
\begin{equation}
\begin{aligned}
\obs{S}\ket{\tau}_c\ket{\psi(\tau)}_r
&=e^{-\frac{i}{\hbar}\obs{T}_{c'}\otimes\obs{H}_{r'}}(\obs{W}_c\otimes\obs{W}_r)\ket{\tau}_c\ket{\psi(0)}_r\\
&=e^{-\frac{i}{\hbar}\obs{T}_{c'}\otimes\obs{H}_{r'}}\ket{\tau}_{c'}\ket{\psi'(0)}_{r'}\\
&=\ket{\tau}_{c'}\ket{\psi'(\tau)}_{r'}.
\end{aligned}
\end{equation}

This proves both the equivalence of the Hamiltonians and the equivalence
of the selected histories. The special case $\obs{H}_{r'}=0$ gives the
trivial PW system used in the spectral discussion above.
Since we proved the equivalence of any ideal PW system with the trivial PW system on the same space, all ideal PW systems are equivalent.
\end{proof}

We see that any two PW systems with equal dimensions are unitarily equivalent \emph{even if there are no interactions between the clock and the rest of the world}.

Now it is easy to see that any two histories of the same PW system are indistinguishable in a bare PW system.

\begin{corollary}
\label{thm:maximal-ambiguity-world}
Under the bare PW equivalence relation, any two histories following the same evolution law are equivalent.
\end{corollary}
\begin{proof}
It follows from Theorem~\ref{thm:maximal-ambiguity}, applied to two PW systems obtained from the same Hamiltonian $\obs{H}=\obs{H}_c\otimes\obs{I}_r+\obs{I}_c\otimes\obs{H}_r$ but with two possibly distinct temporal states $\ket{0}_c\ket{\psi(0)}$ and $\ket{0}_c\ket{\psi'(0)}$.

It may be helpful to see this more explicitly. For any two unit vectors $\ket{\psi(0)}$ and $\ket{\psi'(0)}$ in $\hilbert_r$, there is a unitary transformation $\obs{S}_{0}:\hilbert_r\to\hilbert_r$ so that $\obs{S}_{0}\ket{\psi(0)}=\ket{\psi'(0)}$.
Since $\ket{\psi(\tau)}=e^{-\frac{i}{\hbar}\obs{H}_r\tau}\ket{\psi(0)}$, we get $\ket{\psi'(\tau)}=e^{-\frac{i}{\hbar}\obs{H}_r\tau}\ket{\psi'(0)}=e^{-\frac{i}{\hbar}\obs{H}_r\tau}\obs{S}_{0}\ket{\psi(0)}=e^{-\frac{i}{\hbar}\obs{H}_r\tau}\obs{S}_{0}e^{+\frac{i}{\hbar}\obs{H}_r\tau}\ket{\psi(\tau)}$. Then, by defining $\obs{S}_\tau:=e^{-\frac{i}{\hbar}\obs{H}_r\tau}\obs{S}_{0}e^{+\frac{i}{\hbar}\obs{H}_r\tau}$, we get $\obs{S}_{\tau}\ket{\psi(\tau)}=\ket{\psi'(\tau)}$.
Now construct the operator $\obs{S}:\hilbert\to\hilbert$, $\obs{S}=\int_{\R}\dyad{\tau}_c\otimes\obs{S}_\tau \dd\tau$. Then, $\obs{S}\ket{\tau}_c\ket{\psi(\tau)}_r=\ket{\tau}_c e^{-\frac{i}{\hbar}\obs{H}_r\tau}\obs{S}_{0}e^{+\frac{i}{\hbar}\obs{H}_r\tau}\ket{\psi(\tau)}_r=\ket{\tau}_c e^{-\frac{i}{\hbar}\obs{H}_r\tau}\obs{S}_{0}\ket{\psi(0)}_r=\ket{\tau}_c\ket{\psi'(\tau)}_r$.
Since $\ket{\tau}_c\ket{\psi'(\tau)}_r=e^{-\frac{i}{\hbar}\obs{H}\tau}\ket{0}_c\ket{\psi'(0)}_r$, we also get $\obs{S}\obs{H}=\obs{H}\obs{S}$, establishing the equivalence of the histories under the same law.
\end{proof}

\begin{corollary}
\label{thm:maximal-ambiguity-law}
Two bare PW systems with different Hamiltonians for the rest of the world are indistinguishable.
\end{corollary}
\begin{proof}
Follows from Theorem~\ref{thm:maximal-ambiguity}, applied to the Hamiltonians $\obs{H}=\obs{H}_c\otimes\obs{I}_r+\obs{I}_c\otimes\obs{H}_r$ and $\obs{H}'=\obs{H}_c\otimes\obs{I}_r+\obs{I}_c\otimes\obs{H}_r'$.
\end{proof}

\begin{remark}
\label{rem:ambiguity-way-out}
One may hope to avoid the ambiguity by allowing the clock to be affected by the rest of the world and/or by replacing the projective observable $\obs{T}_c$ with a POVM~\cite{Busch2002Time-energy-uncertainty-relation,LoveridgeMiyadera2019RelativeQuantumTime,HohnSmithLock2021TrinityOfRelationalQuantumDynamics,GemsheimRost2023EmergenceOfTimeFromQuantumInteractionWithTheEnvironment}.
But the space of projective observables is a strict subspace of the space of POVMs, and the space of noninteracting Hamiltonians is a strict subspace of the space of interacting Hamiltonians, so, at least at the level of available decompositions, such generalizations enlarge rather than shrink the space of candidates. Therefore they do not by themselves remove the ambiguity; additional physical restrictions would be needed.
\end{remark}

\subsection{Maximal ambiguity \emph{vs.} clock ambiguity}
\label{s:ambiguity-vs-clock}

Theorems~\ref{thm:maximal-ambiguity-finite-cyclic},~\ref{thm:maximal-ambiguity}, and~\ref{thm:maximal-ambiguity-discrete} show that ideal clocks exhibit the stronger ambiguity. This is stronger than the original Albrecht-Iglesias ambiguity for histories, because it also intertwines the evolution laws. Together, these results show that, even if we impose the noninteraction condition, the ambiguity is maximal whenever the clock is ideal and the timeless state is stationary.

\section{The noninteraction condition and spectral uniformization}
\label{s:reviews-Marletto-Vedral}

Despite the fact that Theorems~\ref{thm:maximal-ambiguity-finite-cyclic},~\ref{thm:maximal-ambiguity}, and~\ref{thm:maximal-ambiguity-discrete} prove the maximal clock ambiguity not only for histories, but also for laws, one may think that even the weaker clock ambiguity was already eliminated by Marletto and Vedral~\cite{MarlettoVedral2017EvolutionWithoutEvolutionAndWithoutAmbiguities}. How can these results be reconciled with that conclusion?

Marletto and Vedral~\cite{MarlettoVedral2017EvolutionWithoutEvolutionAndWithoutAmbiguities} propose that the unitary transformation used by Albrecht and Iglesias to transform a solution into another one is too general, changing the form of the Hamiltonian~\eqref{eq:PW-noninteracting} by introducing an interaction term as in
\begin{equation}
\label{eq:PW-interacting}
\obs{H}=\obs{H}_{c'}\otimes\obs{I}_{r'}+\obs{I}_{c'}\otimes\obs{H}_{r'} + \obs{H}_{\mathrm{int}}
\end{equation}
where $\obs{H}_{\mathrm{int}}$ specifies a nontrivial interaction between the subsystems $c'$ and $r'$.
This suggests that the ambiguity might be removed by allowing only transformations $\obs{M}$ for which $\obs{H}_{\mathrm{int}}=0$.

In their proof, they studied the unitary transformations $\obs{W}$ mapping the old TPS into a new one. They considered that the TPSs obtained from the original one by a transformation $\obs{W}$ that commutes with the total Hamiltonian $\obs{H}$, or by one of whose nonlocal part $\obs{W}_{cr}$ commutes with $\obs{H}$, \emph{``would have a trivial, local action on $\kett{\psi}$''}. They added that in the other cases the Hamiltonian gains an interaction term $\obs{H}_{\mathrm{int}}\neq 0$ in the new TPS.

However, the results above show that the same PW system can even contain different world Hamiltonians in different choices of the TPS. Proposition~\ref{thm:continuous-zero} shows this for the continuous clocks, and Theorem~\ref{thm:maximal-ambiguity} details it.
So there are distinct TPSs in which not only there is no interaction term, but even the world's Hamiltonian can be replaced by any other self-adjoint Hamiltonian on a unitarily isomorphic Hilbert space.
The proof in Marletto and Vedral~\cite{MarlettoVedral2017EvolutionWithoutEvolutionAndWithoutAmbiguities} treats nonlocal changes of TPS as leading to interaction terms in the transformed Hamiltonian. The constructions above show that there is another possibility: a nonlocal pullback of the TPS can leave the Hamiltonian noninteracting, while changing the Hamiltonian assigned to the rest of the world. This is possible because the spectrum of the total Hamiltonian $\obs{H}$ remains unchanged, being $\R$ in all cases.

One natural way to understand the difference is the following. If one assumes that equality of the total Hamiltonian's spectral type also fixes the spectra of the world's Hamiltonian $\obs{H}_r$, then the transformations considered here are invisible. But for ideal clocks this assumption fails, because the clock spectrum uniformizes the rest spectrum.
It is possible that $\varsigma(\obs{H}_c\otimes\obs{I}_r+\obs{I}_c\otimes\obs{H}_r)=\varsigma(\obs{H}_{c'}\otimes\obs{I}_{r'}+\obs{I}_{c'}\otimes\obs{H}_{r'})$ and $\varsigma(\obs{H}_c)=\varsigma(\obs{H}_{c'})$ even if $\varsigma(\obs{H}_r)\neq\varsigma(\obs{H}_{r'})$. In finite dimension $\varsigma(\obs{H}_r)=\varsigma(\obs{H}_{r'})$, but for continuous time, the spectrum is uniformized. And also for discrete time, we should look at the unitary evolution operators, and then the same applies, as seen in Example~\ref{example:two-qubits} and Theorems~\ref{thm:maximal-ambiguity-finite-cyclic} and~\ref{thm:maximal-ambiguity-discrete}.

A charitable reconstruction of the no-ambiguity conclusion is that noninteracting decompositions related by unitary equivalence are being counted as the same physical description. The present results can be understood, from this point of view, as showing that, for ideal clocks, this equivalence relation is too coarse to recover a unique rest dynamics, and in fact the equivalence classes become united into a single class, which contains all possible self-adjoint world Hamiltonians on unitarily isomorphic Hilbert spaces.

\section{Physical meaning \vs ambiguity}
\label{s:physical-meaning}

Since the origin of the clock ambiguity is the high-dimensional unitary symmetry, there is a natural way out of it: reduce the symmetry.
In fact, usual formulations of quantum theory have this symmetry reduction built in, by providing complete information about the observables.
Consider, for example, a scalar wavefunction $\psi$ in the one-dimensional space, satisfying the Schr\"odinger equation
\begin{equation}
\label{eq:schrod-one-dim-scalar}
i\hbar\pdv{t}\psi(x,t)=\[-\frac{\hbar^2}{2m}\pdv[2]{x}+V(x)\]\psi(x,t).
\end{equation}

In the PW formalism, the two Hilbert spaces are $\hilbert_c\cong\hilbert_r\cong L^2(\R)$, the two Hamiltonians are $\obs{H}_c$ and $\obs{H}_r$, and the allowed timeless states have the form
\begin{equation}
\label{eq:PW-one-dim-scalar}
\begin{cases}
\obs{H}=-i\hbar\pdv{\tau}\otimes\obs{I}_r +\obs{I}_c\otimes\[\frac{\widehat{p}_x^2}{2m}+V(\widehat{x})\]\\
\kett{\Psi}=\int_\R\ket{\tau}_c\ket{\psi(\tau)}_r\dd\tau\\
\end{cases}
\end{equation}
where $\widehat{x}$ and $\widehat{p}_x=-i\hbar\pdv{x}$ are the position and momentum observables, and $\ket{\psi(\tau)}_r=\int_\R\psi(x,\tau)\ket{x}_r\dd x$.

The Maximal Ambiguity Theorem~\ref{thm:maximal-ambiguity} implies that this system is unitarily equivalent to a PW system with any other Hamiltonian $\obs{H}_r$ on the Hilbert space $L^2(\R)$, and any solution of the Schr\"odinger equation for $\obs{H}_r$.

Moreover, since all infinite-dimensional separable Hilbert spaces are isomorphic, Theorem~\ref{thm:maximal-ambiguity} implies that any solution $\psi(x,\tau)$ of equation~\eqref{eq:schrod-one-dim-scalar} is unitarily equivalent to any solution $\psi'(x_1,\ldots,x_{dn},\tau)$ of any non-relativistic Schr\"odinger equation for any number $n$ of particles and any space dimension $d$, as in the PW system
\begin{equation}
\label{eq:PW-d-dim-scalar}
\begin{cases}
\obs{H}=-i\hbar\pdv{\tau}\otimes\obs{I}_r\\
\phantom{\obs{H}=}+\obs{I}_c\otimes\[-\sum_{j=1}^n\frac{\hbar^2}{2m_j}\nabla_j^2+V(\obs{r}_1,\ldots,\obs{r}_n)\] \\
\kett{\Psi}=\int_\R\ket{\tau}_c\ket{\psi(\tau)}_r\dd\tau,\\
\end{cases}
\end{equation}
where $\obs{r}_j=\(\widehat{x}_{j1},\ldots,\widehat{x}_{jd}\)$ and $\nabla_j^2=\sum_{k=1}^d\pdv[2]{x_{jk}}$.

This may seem hard to grasp only because we often underappreciate the already surprising fact that all infinite-dimensional separable Hilbert spaces are isomorphic.

To resolve the ambiguity, we can require the transformations of the PW system~\eqref{eq:PW-one-dim-scalar} to preserve additional structures,
\begin{equation}
\label{eq:PW-one-dim-scalar-plus-explicit}
\begin{cases}
\obs{H}=-i\hbar\pdv{\tau}\otimes\obs{I}_r +\obs{I}_c\otimes\[\frac{\widehat{p}_x^2}{2m}+V(\widehat{x})\]\\
\kett{\Psi}=\int_\R\ket{\tau}_c\ket{\psi(\tau)}_r\dd\tau\\
\(\obs{H}_c\otimes\obs{I}_r,\obs{T}_c\otimes\obs{I}_r;\obs{I}_c\otimes\widehat{x},\obs{I}_c\otimes\widehat{p}_x\) \\
\end{cases}
\end{equation}

The PW system~\eqref{eq:PW-one-dim-scalar-plus-explicit} was written perhaps too explicitly as a decomposition $c+r$, so let us replace the local operators $\obs{H}_c$, $\obs{T}_c$, $\widehat{x}$ and $\widehat{p}_x$ with global operators $\obs{C}$, $\obs{T}$, $\obs{X}$ and $\obs{P}$ on $\hilbert$,
\begin{equation}
\label{eq:PW-one-dim-scalar-plus}
\begin{cases}
\obs{H}=\obs{C}+\frac{1}{2m}\obs{P}^2+V(\obs{X}) \\
[\obs{T},\obs{C}]=[\obs{X},\obs{P}]=i\hbar,\\
[\obs{T},\obs{X}]=
[\obs{T},\obs{P}]=
[\obs{C},\obs{X}]=
[\obs{C},\obs{P}]=
0.\\
\end{cases}
\end{equation}

Let us verify that the upgraded PW system allows a complete recovery of the initial quantum theory.
First, the observables $\obs{C}$ and $\obs{T}$ determine an algebra of observables $\mc{A}_c$, consisting of all observables of the form $f\(\obs{T},\obs{C}\)$, where $f:\R^2\to\R$.
Similarly, the observables $\obs{X}$ and $\obs{P}$ determine another algebra of observables $\mc{A}_r$.
Any pair of observables $\obs{A}\in\mc{A}_c$ and $\obs{B}\in\mc{A}_r$ commute, and together, the two algebras generate the full algebra of observables on $\hilbert$.
Then, based on the main result from~\cite{ZanardiLidarLloyd2004QuantumTensorProductStructuresAreObservableInduced}, the two pairs of observables $\(\obs{T},\obs{C}\)$ and $\(\obs{X},\obs{P}\)$ uniquely induce a tensor product structure $\hilbert=\hilbert_c\otimes\hilbert_r$ (modulo the usual domain and regularity assumptions in the infinite-dimensional case).
Then, we can recover $\obs{H}_c$, $\obs{H}_r$, $\obs{T}_c$, $\widehat{x}$, and $\widehat{p}_x$ from their TPS-neutral counterparts by 
$\obs{H}_c\otimes\obs{I}_r=\obs{C}$, $\obs{I}_c\otimes\obs{H}_r=\obs{H}-\obs{C}$, $\obs{T}_c\otimes\obs{I}_r=\obs{T}$, $\obs{I}_c\otimes\widehat{x}=\obs{X}$ and $\obs{I}_c\otimes\widehat{p}_x=\obs{P}$.
By choosing an eigenbasis $\(\ket{x}_r\)_{x\in\R}$ of $\widehat{x}$, we can also recover the wavefunction
\begin{equation}
\label{eq:one-dim-scalar-wavefunction}
\psi(x,\tau)=\braket{x}{\psi(\tau)}_r=\bra{x}\brakett{\tau}{\Psi}=\braakett{\tau,x}{\Psi},
\end{equation}
which encodes in its patterns everything about its world.

The world's Hamiltonian in terms of positions is recovered as $\matrixel{x}{\obs{H}_r}{\psi(\tau)}_r=\[-\frac{\hbar^2}{2m}\pdv[2]{x}+V(x)\]\psi(x,\tau)$, as it appears in equation~\eqref{eq:schrod-one-dim-scalar}.

More generally, we can describe a system of $n$ particles in $d$ space dimensions with $\obs{X}_{jk}:=\obs{I}_c\otimes\widehat{x}_{jk}$ and $\obs{P}_{jk}:=\obs{I}_c\otimes\widehat{p}_{jk}$ by specifying
\begin{equation}
\label{eq:PW-d-dim-scalar-plus}
\begin{cases}
\obs{H}=\obs{C}-\sum_{j=1}^n\frac{\hbar^2}{2m_j}\nabla_j^2+V(\obs{R}_1,\ldots,\obs{R}_n) \\
[\obs{T},\obs{C}]=i\hbar,[\obs{X}_{jk},\obs{P}_{j'k'}]=i\hbar\delta_{jj'}\delta_{kk'},\\
[\obs{T},\obs{X}_{jk}]=
[\obs{T},\obs{P}_{jk}]=
[\obs{C},\obs{X}_{jk}]=
[\obs{C},\obs{P}_{jk}]=
0,\\
\end{cases}
\end{equation}
where $\obs{R}_j=\(\obs{X}_{j1},\ldots,\obs{X}_{jd}\)$ and $\nabla_j^2=\sum_{k=1}^d \frac{\partial^2}{\partial x_{jk}^2}$ in the joint position representation.

By requiring the symmetry transformations of $\hilbert$ to preserve the additional structures, we reduce the symmetry.
Then, symmetry transformations can be interpreted as usual, by making sure to keep track of the correspondence between operators and the physical properties they represent, and to  transform covariantly the operators.
For example, for the system~\eqref{eq:PW-one-dim-scalar-plus}, the wavefunction remains invariant, because $\braket{x}{\psi(\tau)}_r=\braket{e^{-i\obs{W}_r}x}{e^{-i\obs{W}_r}\psi(\tau)}_r$.
The square amplitude $\abs{\psi(x,\tau)}^2=\expval{\widehat{x}}{\psi(\tau)}_r$ is just the conditional expectation value $E(\obs{I}_c\otimes\widehat{x}|\tau)$ of the stationary observable $\obs{I}_c\otimes\widehat{x}$ given that the clock reads $\tau$, defined in~\cite{PageWootters1983EvolutionWithoutEvolution} as in equation~\eqref{eq:exp-val}.

In the PW system~\eqref{eq:PW-one-dim-scalar-plus}, the vector $\kett{\Psi}$ is like any other null eigenvector of $\obs{H}$. 
As a bare unit vector, it does not contain any useful information about the world. The information is extracted from the relation between this vector, the clock, and the position operator $\obs{I}_c\otimes\widehat{x}$. The position operator $\widehat{x}$ admits as position eigenvectors the vectors $\ket{x}$, for all $x\in\R$. These vectors form a basis of $\hilbert_r$ and give the components of $\ket{\psi(\tau)}_r$, \ie the very wavefunction $\psi(x,\tau)=\braket{x}{\psi(\tau)}$. 

In one dimension, the position observable is complete, in the sense that its eigenvectors can be used to distinguish any two states. In three dimensions and $n$ scalar particles, the position operators $\widehat{x}_{j1},\widehat{x}_{j2},\widehat{x}_{j3}$, for $j\in\{1,\ldots,n\}$ form a \emph{complete set of commuting observables}. Their common eigenvectors are able to extract the full information about the state of the system.
However, they are not sufficient to also recover the Hamiltonian $\obs{H}_r$ and the observables of interest.
But if we also include in our set of predefined observables their canonical conjugates, the momentum observables, we obtain a \emph{complete set of generating observables}, \ie a set of observables that generate the algebra $\mc{A}_r$.
In general, any quantum theory needs such a complete set of generating observables, to be able to interpret physically the abstract unit vector $\ket{\psi}$, the observables, and the Hamiltonian.
Then, together with the clock observables $\obs{C}$ and $\obs{T}$ that generate the algebra $\mc{A}_c$, we have a complete set of generating observables for the algebra $\mc{A}=\mc{A}_c\otimes\mc{A}_r$.

Our treatment of scalar particle examples in non-relativistic quantum theory can be easily extended to any quantum theory, by selecting a set of observables that completely generates the algebra of observables on $\hilbert$, allowing the decomposition into clock system and rest of the world and the recovery of the information encoded in $\ket{\psi}_r$.

Therefore, when we specify a quantum system, we also need to specify the operators representing the physical properties of interest, and if we apply a transformation, we need to keep track of this correspondence, by covariantly transforming the operators.

\begin{remark}
\label{rem:nothing-new}
This is, in fact, precisely what has always been done in quantum theory.
Ambiguity arose from ignoring this correspondence, even if this was motivated by the hope that the physical meaning of observables must emerge somehow just from the minimal PW system.
\end{remark}

\section{Ambiguity and symmetry}
\label{s:physical-ambiguity}

Sometimes what looks like ambiguous descriptions may be just equivalent descriptions from different perspectives. Their existence is allowed and even required by the spacetime symmetries.
Spacetime symmetries do not allow the maximal ambiguity shown by Theorem~\ref{thm:maximal-ambiguity}, since the spacetime symmetry group is much smaller than the unitary group of the total Hilbert space, but the same symmetries are also incompatible with the complete absence of ambiguity.

Already in the example of a scalar particle in one dimension, we can see that the position observable $\widehat{x}$ should not be absolute: we should be allowed to choose a different position observable $\widehat{x}'=\widehat{x}-x_0\obs{I}_r$, where $x_0\in\R$, corresponding to a different origin of the one-dimensional space. Both should be understood to represent the same physical property, position.
Similarly, we should be able to use $\widehat{x}'=-\widehat{x}$ and $\widehat{p}'=-\widehat{p}$ as position and momentum.

This freedom is even greater for a scalar particle in three dimensions.
Its wavefunction is an element of $L^2(\R^3)$, and the position observables $\widehat{x}_{j1}$, $\widehat{x}_{j2}$, and $\widehat{x}_{j3}$ decompose this space as a tensor product of the form $L^2(\R^3)\cong L^2(\R)_x\otimes L^2(\R)_y \otimes L^2(\R)_z$. But this tensor product decomposition depends on the directions of the three axes of a particular reference frame. This is in addition to choosing a different origin of the reference frame.

Moreover, in another reference frame in relative motion with respect to the first one, the origin of the old reference frame changes linearly in time, so its position observables will appear to be time-dependent, $\widehat{x}\mapsto\widehat{x}-v\tau\obs{I}_r$~\cite{Weinberg2015LecturesOnQuantumMechanics}.
The position eigenvectors will appear to depend on $\tau$ like $\ket{x}\mapsto e^{-\frac{i}{\hbar}(m v x-\frac{1}{2} m v^2\tau)}\ket{x-v_x\tau}$.
This time-dependence of the position observable can be compensated in the new reference frame by using a different tensor product decomposition $\hilbert_{c'}\otimes\hilbert_{r'}$ which is not equivalent to $\hilbert_{c}\otimes\hilbert_{r}$ under local unitary transformations. One may expect this to happen only in Special Relativity, where space and time are ``mixed'' by Poincar\'e transformations. But already in Galilean relativity, the two reference frames that move relative to each other see inequivalent decompositions $\hilbert_c\otimes\hilbert_r$ and $\hilbert_{c'}\otimes\hilbert_{r'}$, because the Galilei transformations act on the four-dimensional Galilean spacetime $\R^4$, and not separately on space and time, and different reference frames correspond to different decompositions into space and time $\R^4\cong\R^3\times\R$.
The space slices are the same in all inertial reference frames (except for having different and relatively time-dependent origins), but the time directions transverse to these slices in four dimensions differ~\cite{Cartan1923ConnexionAffine,Kunzle1972GalileiSpacetime,Stoica2016DegenerateMetricsAndTheirApplicationsToSpacetime}.

In Special Relativity the decomposition $\hilbert\cong\hilbert_c\otimes\hilbert_r$ is also relative, because the decomposition $\R^4\cong\R^3\times\R$ is relative, and both the time direction and space subspaces are different in different reference frames.
If we consider diffeomorphism invariance, as when dealing with curved spacetime, more ambiguity is required.
In fact, this already happens in Minkowski spacetime if we consider non-inertial reference frames, as the Unruh effect demonstrates~\cite{Unruh1976UnruhEffect}.

Therefore, such an ``ambiguity'' is allowed and even required, arising from changes of reference frame under spacetime isometries. This is why not only is there no unique decomposition $\hilbert_c\otimes\hilbert_r$, but it is even undesirable to remove it completely.

But we should not conclude from this that all clock ambiguities are due to changing the reference frame or the coordinate system.
The spacetime symmetry group is much smaller than the full unitary group of $\hilbert$.
And, in Section~\ref{s:no-accept-ambiguity}, we will see that physical observations reject most of the clock ambiguities.

Therefore, any elimination of the ambiguity has to be \emph{up to physical symmetries}, including spacetime symmetries and, in some formulations of quantum theory, gauge symmetries.

\section{The origin of the clock ambiguity}
\label{s:ambiguity-origin}

We have seen that keeping track of the physical meanings of the observables removes the ambiguity in the PW formalism up to spacetime symmetries.
Then where did the clock ambiguity come from?
It came precisely from ignoring the physical meanings of the observables.
A central reason why Albrecht and Iglesias discuss the clock ambiguity was to eliminate a predefined notion of space, sometimes seen as unnecessary baggage~\cite{Wheeler1983LawWithoutLaw,Wheeler1990InformationPhysicsQuantumLinks,CarrollSingh2019MadDogEverettianism,Carroll2021RealityAsAVectorInHilbertSpace}, and to regard space and locality as emergent and approximate.

\begin{remark}
\label{rem:AI-emergent}
Albrecht~\cite{Albrecht1995TheoryOfEeverythingVsTheoryOfAnything} did not present the clock ambiguity as a rejection of the PW formalism, he rather accepted it and explored the possibility that it could explain physical phenomena like the emergence of classicality, proposing a zero-perturbation inflationary scenario.
Albrecht and Iglesias returned to this theme in~\cite{AlbrechtIglesias2008ClockAmbiguityAndTheEmergenceOfPhysicalLaws} (titled \emph{``Clock ambiguity and the emergence of physical laws''}):
\begin{quote}
the clock ambiguity is a fundamental characteristic of physical laws, which forces us to regard other crucial properties of the physical world such as space, locality, gravity, gauge symmetries, and cosmology as emergent and approximate.
\end{quote}

In~\cite{AlbrechtIglesias2012ClockAmbiguityImplicationsNewDevelopments}, they proposed that the clock ambiguity offers
\begin{quote}
the possibility of predicting the dimensionality of space, and also relate the second derivative of the density of states to the heat capacity of the Universe.
\end{quote}
\end{remark}

We can summarize the hypothesis they explored as
\begin{hypothesis}
\label{hypothesis:emergence}
Space, locality, gravity, gauge symmetries, and other physical structures or properties should not be assumed to be fundamental, but should be derived as emergent properties of the physical world from equation $\obs{H}\kett{\Psi}=0$.
\end{hypothesis}

Hypothesis~\ref{hypothesis:emergence} is the root of the clock ambiguity, which Albrecht and Iglesias tried to interpret as giving just enough freedom to allow for the emergence of the physical laws and various physical structures like space.

As I understand their view, one cannot presume that position observables (or whatever operators define space) have a fundamental existence, because they should emerge from considerations of locality and classicality.
This can be contrasted with standard presentations of quantum mechanics~\cite{vonNeumann1955MathFoundationsQM,Dirac1958ThePrinciplesOfQuantumMechanics}, which specify what operators represent positions. The possible eigenvalues of the position operators define the position configuration space, on which the wavefunction $\psi(x,\tau)=\braket{x}{\psi(\tau)}_r$ propagates. The position configuration space, along with the interaction term in the Hamiltonian, \eg the potential, encodes the dimension of space and locality. But if we get rid of these specifications, we need a way to recover space, locality, and all other physical structures.

Could it be that the physical meaning associated with the operators is redundant, once we know the relational properties of the observables?
Maybe the spectral properties of the operators are sufficient to decode their physical meanings.
For example, maybe an operator with the spectrum $\R$ can only be a position operator in another reference frame.
This is not enough, as we can see from the existence of similar operators of the form $\obs{S}\(\obs{I}_c\otimes\widehat{x}\)\obs{S}^\dagger$, including the momentum operator, but we may hope that, by also taking into account the relations between various operators, including the Hamiltonian, we can get a complete description of reality. An example of such a relation is given by the Hamiltonian from equation~\eqref{eq:PW-one-dim-scalar-plus}, which stands in a particular relation to the momentum and position operators: $\obs{H}_c$ is a momentum operator that generates translations on the $\tau$-line, and $\obs{H}_r=H_r(\widehat{x},\widehat{p}_x)$.
The relations characterizing our world are more complicated, but maybe all of them emerge from just the PW system, as Hypothesis~\ref{hypothesis:emergence} proposes.

The problem is that if we have a set of such operators, for example position, momentum, and Hamiltonian operators satisfying a relation like $\obs{H}_r=H_r(\widehat{x},\widehat{p}_x)$, there are infinitely many other sets of operators in the exact same relations.
They are given by unitary transformations $\obs{S}_r$ of $\hilbert_r$, and have the general form $\widehat{x}'=\obs{S}_r\widehat{x}\obs{S}_r^\dagger$, $\widehat{p}_x'=\obs{S}_r\widehat{p}_x\obs{S}_r^\dagger$, and $\obs{H}_r'=\obs{S}_r\obs{H}_r\obs{S}_r^\dagger$. Then, the functional dependence of $\obs{H}'$ in terms of $\widehat{x}'$, and $\widehat{p}'$ remains unchanged,  $\obs{H}'=H_r(\widehat{x}',\widehat{p}')$.
Also, the commutation relations $[\widehat{x},\widehat{p}]=i\hbar$ remain unchanged, $[\widehat{x}',\widehat{p}']=i\hbar$~\cite{Stoica2022SpaceThePreferredBasisCannotUniquelyEmergeFromTheQuantumStructure}.
But they cannot all be accounted for as the same operators seen in different reference frames, because the spacetime symmetry group is much smaller than the unitary group, even if we limit ourselves to those unitary transformations that commute with the Hamiltonian.

One may think that it is irrelevant whether we choose $\obs{H}_r$, $\widehat{x}$, and $\widehat{p}$, or $\obs{H}_r'$, $\widehat{x}'$, and $\widehat{p}'$, because they satisfy the same relations, so they must also describe one and the same reality.
But in relation to the state vector, these choices lead to different physical situations, for example different distributions of the wavefunction on the position configuration space, so the choice of operators used to represent the physical properties matters.
This is why in Section~\ref{s:physical-meaning} we fixed the correspondence between operators and their physical meanings.

\begin{remark}
\label{rem:not-change-of-basis}
To avoid possible confusions, I have to insist a bit on what we're discussing here.
The point is that there are different ways to choose the operators that represent physical properties, and the way to show this is by starting with some choice, and use unitary symmetries beyond the spacetime and gauge symmetries to obtain different choices satisfying the same relations, but that are not physically equivalent.
I am not talking about representing a particular physical situation using a different basis in $\hilbert$. Of course, to use the same example of a scalar particle in one space dimension, if we know the operators $\obs{H}$, $\widehat{x}$, and $\widehat{p}$, and the state vector $\ket{\psi}$, and we change from the position basis $(\ket{x})_{x\in\R}$ to another basis $(\obs{S}\ket{x})_{x\in\R}$, in the new basis they will have a different form, so that $\braket{\obs{S}x}{\obs{S}\psi}=\braket{x}{\psi}$, but here I am not talking about changing the basis.
Also I do not mean that, when we construct mathematically the state space $\hilbert$, but prior to representing the state as a vector in $\hilbert$, it matters how we choose the position basis, or the position operator.
In such a situation all choices are equally good, as long as once a choice made, we consistently use the unit vector to represent our state in function of the choice of the operators. For example, once we know the basis $(\ket{x})_{x\in\R}$, we can choose the unit vector $\ket{\psi}$ to represent a particular wavefunction $\psi(x)$ as $\ket{\psi}=\int_R\psi(x)\ket{x}\dd x$. Obviously $\ket{\psi}$ depends on the way we chose the position basis in the exact way needed to compensate for the choice.
But here I am only talking about the fact that there are more ways to choose the operators to represent the same physical properties while maintaining their relations.
If we forget what operators represent our physical properties, we cannot recover them, as all other choices will lead to different physical states, represented by different wavefunctions.
Similarly, if we want these operators, or structures like space, to emerge, rather than to be given from start as we did in equation~\eqref{eq:PW-one-dim-scalar-plus}, we will obtain too many solutions, corresponding to different physical situations, but all equally justified by the same relations.
This is one of the ways in which the ambiguity manifests.
\end{remark}

One should note that the expectation that operators do not have a predefined physical meaning, and that this meaning should follow only from the relational properties, is fueled by modern relational and structuralist views on physics: everything about the world should be describable in terms of relations and reducible to relations~\cite{sep-structural-realism,Worrall1989StructuralRealism,Wallace2022StatingStructuralRealismMathematicsFirstApproachesToPhysicsAndMetaphysics}, and even by physicalism, which claims that the physical substance contributes only by relational properties, not by its intrinsic nature~\cite{Carroll2021ConsciousnessAndTheLawsOfPhysics}.
It makes perfect sense, for instance, to assume that two Newtonian billiard balls universes are identical if the positions, momenta, and equations describing the collisions of the balls, are the same, regardless of the material substrate making the balls in the two universes. In one of these universes the balls may be made of steel, and in the other they may be cheese balls, but as long as their interactions follow the same equations, from a relational/structuralist point of view the two universes embody the same physics.
Similarly, it should be all the same whether the wavefunction is a complex wave in an infinite-dimensional ``ocean'' of jelly or of some other substrate, as long as it satisfies the same Schr\"odinger equation.
If we accept this independence of the intrinsic nature of the material substrate, it is natural to expect that from the relations alone we can derive everything there is to know about the world.
In particular, this motivates the hypothesis that, if all relational properties are encoded in the bare PW system, this should contain everything we observe in our universe, even our own sentient experiences~\cite{Carroll2021RealityAsAVectorInHilbertSpace,Carroll2021ConsciousnessAndTheLawsOfPhysics}.
This seems to lead to the belief that the operators and their spectral properties are enough to recover space, its dimensions, and everything else, from the bare PW system alone.

But we've seen in Section~\ref{s:ambiguity} that any two PW systems are unitarily equivalent, regardless of the number of dimensions of space, as long as the worlds' Hilbert spaces have the same dimension.
In particular, any two nonrelativistic PW systems are unitarily equivalent regardless of the dimension of space and the number and types of particles.
Even if the Hamiltonian $\obs{H}_r$ is not of the form~\eqref{eq:PW-d-dim-scalar-plus}, but of any other possible form, Theorem~\ref{thm:maximal-ambiguity} shows that any two worlds for which dimension of $\hilbert_r$ is the same are unitarily equivalent.
Therefore, the dimension of space, locality, and other physical features, cannot emerge from a bare PW system as in Definition~\ref{def:PW-system}, but they can be specified by their physical meanings, as shown in Section~\ref{s:physical-meaning}.

Albrecht and Iglesias explored the consequences of a relational stance close to Hypothesis~\ref{hypothesis:emergence}. Marletto and Vedral responded within the same general setting, by asking whether an additional relational constraint, noninteraction, removes the ambiguity. Theorems~\ref{thm:maximal-ambiguity-finite-cyclic},~\ref{thm:maximal-ambiguity}, and~\ref{thm:maximal-ambiguity-discrete} continue this comparison within that setting, but show that the ambiguity becomes maximal unless the physical meanings of observables are specified.

By showing that Hypothesis~\ref{hypothesis:emergence} leads to problems, the Maximal Ambiguity Theorem~\ref{thm:maximal-ambiguity} implies that this hypothesis should be rejected, so we should adopt a solution as in Section~\ref{s:physical-meaning}.
But one may think that the clock ambiguity problem is not a real problem and that, while there are many ways to associate physical properties to the operators, all these ways may be supported simultaneously, saving thus Hypothesis~\ref{hypothesis:emergence}. This hypothesis will be analyzed in Section~\ref{s:no-accept-ambiguity}.

\section{Can we live with the ambiguity?}
\label{s:no-accept-ambiguity}

Recall from Section~\ref{s:physical-meaning} that spacetime symmetries, and diffeomorphism invariance if we include gravity, require a certain degree of ambiguity.
This may suggest that all descriptions of the world allowed by the clock ambiguity are equally valid, being just different perspectives on one and the same reality.
One may hope that, even if there is ambiguity beyond the physical symmetries, the resulting different descriptions of the world are just different but equivalent perspectives on the same world, and are all equally valid.
In such a picture, $\kett{\Psi}$ is not a ``timeless'' state, but a ``clock-neutral'' one, a reference system-neutral state, \ie \emph{``a description of
physics prior to having chosen a temporal reference frame
relative to which the other degrees of freedom evolve''}~\cite{HohnSmithLock2021TrinityOfRelationalQuantumDynamics}.
So it may be tempting to entertain the following hypothesis:
\begin{hypothesis}
\label{hypothesis:embracing-ambiguity}
There are many different ways to choose a clock and to associate physical properties to the operators in a PW system, but all these ways are just different but equivalent perspectives on one and the same reality.
\end{hypothesis}

To test this idea, I use the account of quantum observations in a PW system based on records in the instantaneous state $\ket{\tau}_c\ket{\psi(\tau)}_r$.
Recall Kucha\v{r}'s criticism that the PW proposal does not give an answer to questions about dynamics that refer to observables at different times: \emph{``your interpretation prohibits the time to flow and the system to move!''}~\cite{Page1994ClockTimeAndEntropy,Kuchar1992TimeAndInterpretationsOfQuantumGravity}. Don Page replied that the answer is given by the original inspiration for the PW proposal, Everett's relative states, extended to include many ``now'''s as well as many worlds~\cite{Page1989TimeAsAnInaccessibleObservable},
\begin{quote}
We can never directly test what happened yesterday, but we can check the consequences that a hypothetical scenario for yesterday has on the situation today.
\end{quote}

Then, the present state contains records of the past, and since a quantum measurement can result in different records of the past measurement in the present state, this can be used to explain both the statistics of a sequence of quantum measurements, and our perception of change.
Subsequent literature sometimes saw Kucha\v{r}'s criticism as fatal to the PW proposal, and a series of refinements of the solution based on records followed: Gambini \etal~\cite{GambiniPortoPullin2004RelationalSolutionProblemTimeQuantumMechanicsQuantumGravity} used Rovelli's
evolving constants of the motion~\cite{Rovelli1991TimeInQuantumGravityAnHypothesis,Hajicek1991CommentOnTimeInQuantumGravityAnHypothesis,Rovelli1991QuantumEvolvingConstantsReplyToCommentOnTimeInQuantumGravityAnHypothesis}, Dolby gave an account based on consistent histories~\cite{Dolby2004TheConditionalProbabilityInterpretationOfTheHamiltonianConstraint}, Giovannetti \etal used von Neumann measurements and interactions between the observed system and the measurement device to derive the Born rule~\cite{GiovannettiLloydMaccone2015QuantumTime}, and H\"ohn \etal derived it relationally~\cite{HohnSmithLock2021TrinityOfRelationalQuantumDynamics}.
While each of these solutions formalized this aspect of the PW proposal based on different justifications, I think Page's reply already contains the answer: use the present records of past measurements to derive the probabilities. He invoked Everett's relative states for two reasons. The first is that Everett~\cite{Everett1973TheTheoryOfTheUniversalWaveFunction,Everett1973TheTheoryOfTheUniversalWaveFunction}, who appealed to decoherence before it was widely researched (like Bohm did a few years before~\cite{Bohm1952SuggestedInterpretationOfQuantumMechanicsInTermsOfHiddenVariables}), explains how it leads to records that distinguish the branches. At the same time, central to the PW proposal is to invoke the relative states to account for the world's state at different times similar to how Everett accounted for the branches.
Observing the present at a time $\tau$ does not collapse the timeless state $\kett{\Psi}$ into $\ket{\tau}_c\ket{\psi(\tau)}_r$, for the same reason why observing your own Everettian branch does not collapse the wavefunction to that branch. Moreover, the records of the past contained in the state vector $\ket{\psi(\tau)}_r$ condition the probabilities of those past events.
So Page not only mentioned that the present contains records of the past observables, but the probability should be conditioned both by $\ket{\tau}_c$ and by the projector corresponding to the branch containing the records with the particular values of interest.
For what follows, the reader may consider Page's account of records, one of its more elaborate and formal forms cited above, or Marletto and Vedral's account of how the past history leads to records in the present state~\cite{MarlettoVedral2017EvolutionWithoutEvolutionAndWithoutAmbiguities}.

With the idea of records in mind, consider a system containing records of events taking place in the environment, which will thereby be called ``recorder''. The recorder may even be an observer holding in her brain records of past events, including from the immediate past, for example objects and their properties she just observed in her surroundings.
Some of the records will still be valid, in the sense that the values of the properties of the environment they represent did not change.
For example, a map keeps records of the coordinates of many locations on Earth that, due to the slow changes in geography, remain valid for a long time.
With this in mind, let us express the total state as
\begin{equation}
\label{eq:memory-environment}
\ket{\tau}_c\underbrace{\ket{\tn{recorder}(\tau)}_m\ket{\tn{environment}(\tau)}_e}_{\ket{\psi(\tau)}_r},
\end{equation}
where $\ket{\tn{recorder}(\tau)}_m$ represents the state of the recorder containing records and $\ket{\tn{environment}(\tau)}_e$ represents the rest of the world, including external objects described by the records in the recorder, at time $\tau$.

According to Corollary~\ref{thm:maximal-ambiguity-world}, if we accept the clock ambiguity (in any version), the same PW system equally encodes infinitely many other states of the form
\begin{equation}
\label{eq:memory-environment-other}
\ket{\tau}_c\underbrace{\ket{\tn{recorder}(\tau)}_m\ket{\tn{environment}'(\tau)}_e}_{\ket{\psi'(\tau)}_r},
\end{equation}
where $\ket{\tn{environment}'(\tau)}_e$ can be a completely different state from $\ket{\tn{environment}(\tau)}_e$ , with objects having completely different properties than those recorded in the state $\ket{\tn{recorder}(\tau)}_m$.

Let us focus first on a single property of the environment, whose value is registered by the recorder. 
So we'll consider a PW system in which the state $\ket{\psi(\tau)}_r$ is
\begin{equation}
\label{eq:memory-environment-apple-red}
\ket{\tau}_c\ket{\tn{``the apple is red''}}_m\ket{\tn{the apple is red}}_e,
\end{equation}
and a bare PW system in which the state $\ket{\psi'(\tau)}_r$ is
\begin{equation}
\label{eq:memory-environment-apple-blue}
\ket{\tau}_c\ket{\tn{``the apple is red''}}_m\ket{\tn{the apple is blue}}_e.
\end{equation}

Since the clock ambiguity implies that we can decode from the same PW system alternative states containing different environment states, like in equation~\eqref{eq:memory-environment-apple-blue}, for all possible colors of the apple, it follows that there is no correlation between the color of the apple in the environment and its record maintained in the recorder.
Since this is true of all properties registered by the recorder, it follows that the state of our recorder is completely uncorrelated with the state of the environment.
In particular, if the recorder is an observer Alice, Alice would be unable to know basic facts about the environment, that the color of the apple is red, or even that it is an apple and not some other object, and so on.
She would know nothing about the environment, simply because she can be any copy of $\ket{\tn{recorder}(\tau)}_m$ surrounded by any environment $\ket{\tn{environment}'(\tau)}_e$ that can be decoded from the same state vector.
This is not what we observe in reality, but it is the prediction of Hypothesis~\ref{hypothesis:embracing-ambiguity}.
Therefore, the simple fact that we can know a great deal about our surroundings refutes Hypothesis~\ref{hypothesis:embracing-ambiguity}.
And this is an understatement: Hypothesis~\ref{hypothesis:embracing-ambiguity} leads straight to absurd implications that would make it impossible for us not only to do science, but even to survive and function properly for a fraction of a second. So we cannot live with the ambiguity, literally.

This implies that we cannot simply embrace the ambiguity. The way to remove it so that any transformation is allowed is by keeping track covariantly of what physical property each operator represents.
For example, if we use the ``trinity'' synthesis~\cite{HohnSmithLock2021TrinityOfRelationalQuantumDynamics}, the clock-neutral picture provided by $\kett{\Psi}$ must come with the specification of the physical meaning for a complete set of generating observables. Only then can the fully specified clock-neutral PW system be used as a Rosetta Stone to translate between various choices of the clock, provided that we track the observables covariantly in our symmetry reductions or transformations.
By failing to establish and to track covariantly this correspondence, but wanting it to emerge purely from relations, we only get the clock ambiguity noticed by Albrecht and Iglesias. One may also note that, by specifying the physical properties represented by the operators, we implicitly specify the class of physically equivalent perspectives, as in equation~\eqref{eq:PW-one-dim-scalar-plus}.
This specification reduces the symmetry so that the ambiguity flies under the radar of what we can distinguish physically.

\section{Conclusion}
\label{s:conclusion}

In this article, I reexamined the clock ambiguity problem~\cite{Albrecht1995TheoryOfEeverythingVsTheoryOfAnything,AlbrechtIglesias2008ClockAmbiguityAndTheEmergenceOfPhysicalLaws,AlbrechtIglesias2012ClockAmbiguityImplicationsNewDevelopments} in the Page-Wootters proposal~\cite{PageWootters1983EvolutionWithoutEvolution,Wootters1984TimeReplacedByQuantumCorrelations,Page1986DensityMatrixOfTheUniverse,Page1989TimeAsAnInaccessibleObservable,Page1994ClockTimeAndEntropy}.
I proved that the clock ambiguity extends to include the Hamiltonians and it is maximal for continuous and discrete ideal clocks, even if we impose the noninteraction condition.
I clarified that it is not desirable to completely remove the ambiguity, because spacetime symmetries require a definite freedom of the clock-world decomposition (Section~\ref{s:physical-ambiguity}).
However, accepting the full extent of the clock ambiguity as perspectival is a no-go too, because then nobody would be able to know anything about the external world (Section~\ref{s:no-accept-ambiguity}).
The way out of the ambiguity is to give up the assumptions leading to it (Hypotheses~\ref{hypothesis:emergence} and~\ref{hypothesis:embracing-ambiguity}), and to accept that there is a definite association between operators and the physical properties they represent, even if infinitely many other choices of the operators satisfy the exact same relations (Section~\ref{s:physical-meaning}).
Thus the clock ambiguity returns, extended to Hamiltonians and maximal even under the noninteraction condition, only when the PW structure is treated as a bare description, because such a description is incomplete. Once the physical meaning of the observables is included, the ambiguity is extinguished up to the physical symmetries.

\appendix

\section{Maximal ambiguity for infinite and continuous time}
\label{s:ambiguity-infinite-continuous}

\subsection{Maximal ambiguity for infinite discrete time}
\label{s:ambiguity-discrete-time-infinite}

Now I will show that any ideal discrete PW system has the clock ambiguity problem for histories and laws.

\begin{theorem}[Maximal ambiguity, discrete time]
\label{thm:maximal-ambiguity-discrete}
Any ideal discrete PW system has the clock ambiguity problem for histories and laws.
That is, for any other ideal discrete PW system with $\dim\hilbert_{r'}=\dim\hilbert_r$ and for any temporal states $\ket{\Psi(0)}\in\hilbert$ and $\ket{\Psi'(0)}\in\hilbert'$, there is a unitary operator $\obs{S}:\hilbert_c\otimes\hilbert_r\to\hilbert_{c'}\otimes\hilbert_{r'}$ so that 
\begin{enumerate}
	\item 
the histories of the two PW systems are unitarily equivalent, \ie for any $\tau\in\Z$,
\begin{equation}
\label{eq:PW-unitarily-equivalent-state-discrete}
\ket{\Psi'(\tau)}=\obs{S}\ket{\Psi(\tau)}
\end{equation}
	\item 
and the two evolution laws are unitarily equivalent
\begin{equation}
\label{eq:PW-unitarily-equivalent-H-discrete}
\obs{U}'=\obs{S}\obs{U}\obs{S}^{-1}.
\end{equation}
\end{enumerate}
\end{theorem}
\begin{proof}[Informal proof]
Let $\obs{U}=\obs{U}_c\otimes\obs{U}_r$ and $\obs{U}'=\obs{U}_{c'}\otimes\obs{U}_{r'}$ be the evolution operators of the two PW systems.
It may seem that the fact that $\obs{U}_r$ and $\obs{U}_{r'}$ can have different spectra makes the existence of a unitary equivalence as in equation~\eqref{eq:PW-unitarily-equivalent-H-discrete} impossible, but this is not the case.
Since the spectrum of $\obs{U}_c$ is $\varsigma(\obs{U}_c)=e^{-\frac{i}{\hbar}\R}=e^{-\frac{i}{\hbar}[0,2\pi)}$, the unit circle, we get that $\varsigma(\obs{U})=\varsigma(\obs{U}_c)\times\varsigma(\obs{U}_r)=\varsigma(\obs{U}_c)$, with uniform multiplicity equal to $\dim\hilbert_r$.
The homogenization of the spectrum washes out the spectrum of the unitary evolution operator $\obs{U}_r$, because $e^{-\frac{i}{\hbar}\varphi_r} e^{-\frac{i}{\hbar}\R}=e^{-\frac{i}{\hbar}\R}$ for any $\varphi_r\in\R$.
Then, since $\obs{U}$ and $\obs{U}'$ have the same spectrum, with the same spectral measure class and multiplicity, there is at least one unitary operator $\obs{S}$ as in equation~\eqref{eq:PW-unitarily-equivalent-H-discrete}.
Moreover, from equation~\eqref{eq:PW-WDW-discrete}, $\kett{\Psi}$ is an eigenvector of $\obs{U}$ with eigenvalue $1$, and so is $\kett{\Psi'}$ for $\obs{U}'$.
Therefore, we can choose $\obs{S}$ to satisfy both conditions~\eqref{eq:PW-unitarily-equivalent-H-discrete} and~\eqref{eq:PW-unitarily-equivalent-state-discrete}.
\end{proof}

This explains less formally why no spectral obstruction remains. The following proof makes the shift structure explicit and proves the theorem, even when $\obs{U}_r$ is $\tau$-dependent. In the $\tau$-dependent case, $\obs{U}'$ is no longer of the form $\obs{U}_{c'}\otimes\obs{U}_{r'}$, but of the form $\obs{U}'=\sum_{\tau\in\Z}\dyad{\tau}_{c'}\otimes\obs{U}_{r'}(\tau+1,\tau)$.
\begin{proof}
For  $\tau$-dependent unitary evolution operators, we can no longer use the notation $\obs{U}^\tau$, we will have to use explicitly $\obs{U}(\tau_2,\tau_1)$ and $\obs{U}_r(\tau_2,\tau_1)$.
Let $\(\ket{e_1}_r, \ket{e_2}_r,\ldots\)$ be an orthonormal basis of $\hilbert_r$.
For each $\tau$ and $k$, define
\begin{equation}
\label{eq:PW-tk-shift}
\ket{\tau,k}:=\obs{U}(\tau,0)\ket{0}_c\ket{e_k}_r=\ket{\tau}_c \obs{U}_r(\tau,0)\ket{e_k}_r.
\end{equation}

Then, for all $k,k'$ and all $\tau,\tau'\in\Z$,
\begin{equation}
\label{eq:PW-tk-ortho-discrete}
\braket{\tau,k}{\tau',k'}=\delta_{kk'}\delta_{\tau\tau'}.
\end{equation}

Let $\hilbert_\tau$ be the Hilbert space spanned by the vectors $\ket{\tau,k}$, for fixed $\tau$ and all $k$.
Then, from equation~\eqref{eq:PW-tk-ortho-discrete}, $\hilbert_0$ is a \emph{wandering subspace}, that is, $\obs{U}(\tau,0)\hilbert_0$ and $\obs{U}(\tau',0)\hilbert_0$ are orthogonal whenever $\tau\neq\tau'$, and $\hilbert=\bigoplus_{\tau\in\Z}\obs{U}(\tau,0)\hilbert_0$ (see~\cite{NagyFoiasBercoviciKerchy2010HarmonicAnalysisOfOperatorsOnHilbertSpace} page 2).
Then, as shown for example in~\cite{NagyFoiasBercoviciKerchy2010HarmonicAnalysisOfOperatorsOnHilbertSpace} page 5, $\obs{U}$ is a bilateral shift on $\hilbert$, and it is therefore determined only by the dimension of $\hilbert_r$ up to unitary equivalence.
Therefore, all ideal discrete PW systems with the same dimension are unitarily equivalent to each other and have the clock ambiguity problem for histories and laws.
\end{proof}

Let us see this more explicitly and self-contained, avoiding references to wandering subspaces~\cite{NagyFoiasBercoviciKerchy2010HarmonicAnalysisOfOperatorsOnHilbertSpace}.

\begin{proof}[Alternative proof]
Let $\hilbert_k$ be the Hilbert space spanned by the vectors $\ket{\tau,k}$ defined in equation~\eqref{eq:PW-tk-shift}, for all $\tau\in\Z$ and fixed $k$. Then, from equations~\eqref{eq:PW-tk-shift} and~\eqref{eq:PW-tk-ortho-discrete} we get that 
the restriction of the total time evolution operator $\obs{U}(\tau,0)$ to $\hilbert_k$ coincides with a shift with $\tau\mapsto\tau+1$. 

Consider a \emph{trivial ideal discrete PW system}, \ie a PW system $\hilbert'=\hilbert_{c'}\otimes\hilbert_{r'}$ with unitary evolution $\obs{U}_{r'}=\obs{I}_{r'}$,
\begin{equation}
\label{eq:total-space-shift}
\obs{U}'\ket{\tau}_{c'}\ket{v_k}_{r'}=\ket{\tau+1}_{c'}\ket{v_k}_{r'}.
\end{equation}

We will show that our PW system is unitarily equivalent to the trivial one.
We define the transformation
\begin{equation}
\label{eq:unitary-equiv-trivial-shift}
\obs{S}\ket{\tau,k}=\ket{\tau}_{c'}\ket{v_k}_{r'},
\end{equation}
where $\(\ket{v_k}_{r'}\)_k$ is an orthonormal basis of $\hilbert_{r'}$.
Then, from equations~\eqref{eq:total-space-shift} and~\eqref{eq:unitary-equiv-trivial-shift}, $\obs{U}'\obs{S}=\obs{S}\obs{U}$.
From this, the total evolution operator $\obs{U}$ is unitarily equivalent to $\obs{U'}$ on $\hilbert'$ from equation~\eqref{eq:total-space-shift}.

Moreover, $\obs{S}$ maps each history $(\obs{U}(\tau,0)\ket{0}_c\ket{e_k}_r)_{\tau\in\Z}$ of the first discrete PW system into a history $(\obs{U}_{c'}^{\tau} \ket{0}_{c'}\ket{v_k}_{r'})_{\tau\in\Z}$ of the trivial discrete PW system, as seen in equation~\eqref{eq:unitary-equiv-trivial-shift}.
Therefore, any ideal discrete PW system is unitarily equivalent to a discrete trivial PW system, and so any two ideal discrete PW systems with the same dimension are unitarily equivalent.
\end{proof}


\end{document}